\documentclass{article}

\usepackage{enumerate,subfigure,latexsym} 

\usepackage{amsfonts,amssymb,amsmath,amsthm}

\usepackage[numbers]{natbib}

\usepackage{url}

\makeatletter
\def\blfootnote{\xdef\@thefnmark{}\@footnotetext}
\makeatother
\usepackage{graphicx}
\newcommand\idagger{\rotatebox[origin=c]{180}{\ensuremath{\dagger}}}

\makeatletter
\newif\if@restonecol
\makeatother

\usepackage[algoruled,vlined]{algorithm2e}
\SetKwInOut{Input}{Input}
\SetKwInOut{Output}{Output} 

\usepackage{tabularx, multirow, subfigure, rotating}
\usepackage{ifdraft}

\usepackage[usenames,dvipsnames, table]{xcolor}
\ifdraft{\usepackage[normalem]{ulem}}

{\bfseries}{\rmfamily}

\newtheorem{example}{Example}
\newtheorem{proposition}{Proposition}
\newtheorem{definition}{Definition}
\newtheorem{lemma}{Lemma}
\newtheorem{theorem}{Theorem}

\newcommand\refcite[1]{\citealp{#1}} 
\newcommand\citess[1]{\textsuperscript{\textup{\citealp{#1}}}} 

\def\CC{{\mathbb C}} 
\def\NN{{\mathbb N}} \def\QQ{{\mathbb Q}}
\def\RR{{\mathbb R}} \def\ZZ{{\mathbb Z}} 
\def\A{{\mathcal A}} \def\R{{\mathcal R}} 

\newcommand{\Aff}{\mbox{Aff}}

\newcommand{\CH}{\mbox{CH}}

\usepackage{color}

\begin{document}

\markboth{I.\,Z.~Emiris, V. Fisikopoulos, C. Konaxis \& L. Pe{\~n}aranda}
{An Algorithm for Projections of Resultant Polytopes}

\title{An Oracle-based, Output-sensitive Algorithm\\
for Projections of Resultant Polytopes}
\author{Ioannis Z.~Emiris\thanks{Department of Informatics \&
        Telecommunications, University of Athens, Athens, 15784, Greece.
        {\tt emiris@di.uoa.gr}, {\tt vfisikop@di.uoa.gr}}
        \and
        Vissarion Fisikopoulos\footnotemark[1]
        \and
        Christos Konaxis\thanks{Archimedes Center for Modeling, Analysis \&
        Computation (ACMAC), University of Crete, Heraklio, 71409, Greece.
        {\tt ckonaxis@acmac.uoc.gr}}
        \and
        Luis Pe{\~n}aranda\textsuperscript{\idagger}
}
\date{}

\maketitle
\blfootnote{\hspace{-3.8pt}\textsuperscript{\idagger}IMPA -- Instituto
Nacional de Matem{\'a}tica Pura e Aplicada, Rio de Janeiro, 22460-320,
Brazil.
{\tt luisp@impa.br}}

\begin{abstract}
We design an algorithm to compute the Newton polytope
of the resultant, known as resultant polytope, or its 
orthogonal projection along a given direction.
The resultant is fundamental in algebraic elimination, optimization,
and geometric modeling.
Our algorithm exactly computes vertex- and halfspace-representations
of the polytope using an oracle producing resultant vertices in a
given direction,
thus avoiding walking on the polytope whose dimension is $\alpha-n-1$,
where the input consists of $\alpha$ points in $\ZZ^{n}$.
Our approach is output-sensitive as it makes one
oracle call per vertex and facet.
It extends to any
polytope whose oracle-based definition is advantageous, such as
the secondary and discriminant polytopes.
Our publicly available implementation uses the experimental
CGAL package {\tt triangulation}.
Our method computes $5$-, $6$- and $7$-dimensional polytopes
with $35$K, $23$K and $500$ vertices, respectively, within $2$hrs,
and the Newton polytopes of many important surface equations encountered in
geometric modeling in $<1$sec,
whereas the corresponding secondary polytopes are intractable.
It is faster than tropical geometry software up to dimension $5$ or $6$.
Hashing determinantal predicates accelerates execution up to $100$ times.  
One variant computes inner and outer approximations with, respectively,
90\% and 105\% of the true volume, up to $25$ times faster.

\paragraph{Keywords: } General Dimension, Convex Hull, Regular
Triangulation, Secondary Polytope, Resultant, CGAL Implementation,
Experimental Complexity.
\end{abstract}

\section{Introduction}

Given pointsets $A_0,\dots,A_n\subset \ZZ^n$, we define the pointset
\begin{equation}\label{EQ:Cayley}
\A:=\bigcup_{i=0}^{n} (A_{i} \times \{e_{i}\}) \subset \ZZ^{2n},
\end{equation}
where $e_0,\ldots,e_n$ form an affine basis of $\RR^n$:
$e_{0}$ is the zero vector,\linebreak
\(e_i = (0, \dots, 0, 1, 0, \dots, 0), i = 1, \dots, n\).
Clearly, $|\A| = |A_0|+ \cdots +|A_n|$, where $|\cdot|$ denotes cardinality.
By Cayley's trick (Proposition~\ref{P:Cayley_trick}) the
regular tight mixed subdivisions of the Mink\-owski sum $A_0+ \cdots +A_n$  
are in bijection with the regular triangulations of $\A$, which
are in bijection with the vertices of the {\em secondary polytope} $\Sigma(\A)$
(see Section~\ref{Scombinatorics}).

The {\em Newton polytope} of a polynomial is 
the convex hull of its {\em support}, i.e.\
the exponent vectors of monomials with nonzero coefficient.
It subsumes the notion of degree for sparse multivariate polynomials by
providing more precise information (see Figures~\ref{fig:NewPol}
and~\ref{Fbuchberger}).
Given $n+1$ polynomials in $n$ variables, with fixed supports $A_i$
and symbolic coefficients, their {\em sparse (or toric) resultant} $\R$
is a polynomial in these coefficients which vanishes exactly when the
polynomials have a common root (Definition~\ref{Dresultant}).
The resultant is the most fundamental tool in elimination theory,
it is instrumental in system solving and optimization, and is crucial
in geometric modeling, most notably for
changing the representation of parametric hypersurfaces to implicit.

The Newton polytope of the resultant $N(\R)$, or \textit{resultant polytope},
is the object of our study; it is of dimension $|\A| -2n-1$
(Proposition~\ref{Psummand_dimRes}).
We further consider the case when some of the input coefficients are not
symbolic, hence we seek an orthogonal projection of the resultant polytope.
The lattice points in $N(\R)$ yield a superset of the support of $\R$;
this reduces implicitization~\citess{EmKaKoLB,StuYu08} and computation of
$\R$ to sparse interpolation (Section~\ref{Scombinatorics}).  
The number of coefficients of the $n+1$ polynomials ranges from $O(n)$
for sparse systems,
to $O(n^d d^n)$, where $d$ bounds their total degree.
In system solving and implicitization, one computes $\R$ when all but
$O(n)$ of the coefficients are specialized to constants, hence the
need for resultant polytope projections.  

The resultant polytope is a Minkowski summand of $\Sigma(\A)$,
which is also of dimension $|\A| -2n-1$.
We consider an equivalence relation defined on the $\Sigma(\A)$
vertices, where the classes are in bijection with the vertices of 
the resultant polytope.
This yields an oracle
producing a resultant vertex in a given direction, thus avoiding to compute
$\Sigma(\A)$, which typically has much more vertices than $N(\R)$.
This is known in the literature as an {\em optimization} oracle since
it optimizes inner product with a given vector over the (unknown) polytope.

\begin{figure*}[t] \centering
 \includegraphics[width=.8\textwidth]{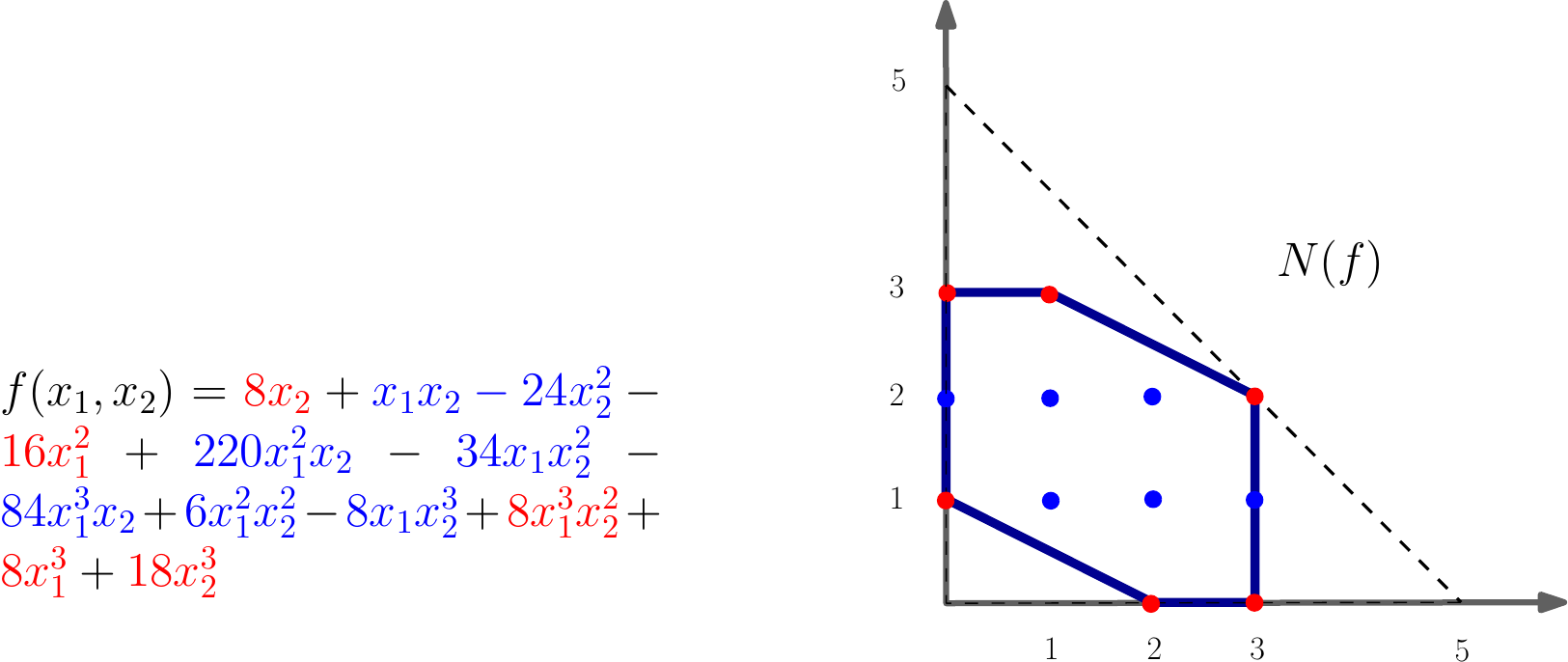}
 \caption{The Newton polytope of a polynomial of degree $5$ in two variables.
Every monomial corresponds to an integral point on the plane. 
The dashed triangle is the corresponding polytope of the dense polynomial of
degree $5$. 
\label{fig:NewPol}} 
\end{figure*}

\begin{example}\label{ExamBicubic}{\rm [The bicubic surface]}
A standard benchmark in geometric modeling is the implicitization of the
bicubic surface, with $n=2$,
defined by $3$ polynomials in two parameters.
The input polynomials have supports
$A_i\subset\ZZ^2,\,i=0,1,2$, with cardinalities $7,6, 14$, respectively;
the total degrees are $3,3,6$, respectively. 
The Cayley set $\A\subset\ZZ^4$, constructed as in Equation~\ref{EQ:Cayley}, has
$7+6+14=27$ points. It is depicted in the following matrix, with coordinates as
columns, where the supports from different polynomials and the Cayley
coordinates are distinguished.
By Proposition~\ref{Psummand_dimRes} it follows that $N(\R)$ 
has dimension $|\A|-4-1=22$; it lies in $\RR^{27}$.

\setlength{\tabcolsep}{0pt} 
\newcommand{\w}{3.8mm} 
\[
\begin{tabular}{|m{.5mm} >{\centering}m{\w} >{\centering}m{\w}
>{\centering}m{\w} >{\centering}m{\w} >{\centering}m{\w} >{\centering}m{\w}
>{\centering}m{\w}| 
>{\centering}m{\w} >{\centering}m{\w} >{\centering}m{\w}
>{\centering}m{\w} >{\centering}m{\w} >{\centering}m{\w}|
>{\centering}m{\w}
>{\centering}m{\w} >{\centering}m{\w} >{\centering}m{\w} >{\centering}m{\w}
>{\centering}m{\w} >{\centering}m{\w} >{\centering}m{\w} >{\centering}m{\w}
>{\centering}m{\w} >{\centering}m{\w} >{\centering}m{\w} >{\centering}m{\w}
>{\centering}m{\w} m{.5mm}| l}

\cline{1-1}
\cline{29-29}

& \color{Maroon} $0$ & \color{Maroon}  $0$ & \color{Maroon}  $1$ &
\color{Maroon} $0$ & \color{Maroon} $2$ & \color{Maroon} $0$ & \color{Maroon}
$3$ &

\color{Blue} $0$ & \color{Blue} $0$ & \color{Blue} $1$ & \color{Blue}  $2$ &
\color{Blue}  $0$ & \color{Blue} $3$ &

 \color{ForestGreen} $0$ & \color{ForestGreen} $0$ & \color{ForestGreen} $1$ &
\color{ForestGreen} $0$ &
\color{ForestGreen} $1$ & \color{ForestGreen} $2$ & \color{ForestGreen}  $1$ &
\color{ForestGreen}  $2$ &
\color{ForestGreen} $1$ & \color{ForestGreen} $2$ & \color{ForestGreen} $3$ &
\color{ForestGreen} $2$ &
\color{ForestGreen} $3$ & \color{ForestGreen} $3$ & \color{ForestGreen} &
        \multirow{2}{*}{ \Big\} \text{support}}\\

& \color{Maroon} $0$ & \color{Maroon}  $1$ & \color{Maroon}  $0$ &
\color{Maroon} $2$ &
\color{Maroon} $0$ & \color{Maroon} $3$ & \color{Maroon} $0$ & 

\color{Blue} $0$ & \color{Blue} $1$ & \color{Blue} $0$ & \color{Blue}  $0$ &
\color{Blue}  $3$ & \color{Blue} $0$ &

 \color{ForestGreen} $0$ & \color{ForestGreen} $1$ &
\color{ForestGreen} $0$ & \color{ForestGreen} $2$ & \color{ForestGreen} $1$ 
& \color{ForestGreen} $0$ & \color{ForestGreen}  $2$ & \color{ForestGreen}  $1$
& \color{ForestGreen} $3$ &
\color{ForestGreen} $2$ & \color{ForestGreen} $1$ & \color{ForestGreen} $3$ &
\color{ForestGreen} $2$ &
\color{ForestGreen} $3$ & \color{ForestGreen} &\\

\cline{2-28}

& \color{BurntOrange} $0$ & \color{BurntOrange}  $0$ & \color{BurntOrange}  $0$
& \color{BurntOrange} $0$ &
\color{BurntOrange} $0$ & \color{BurntOrange} $0$ & \color{BurntOrange} $0$ &
\color{BurntOrange} $1$ &
\color{BurntOrange} $1$ 
& \color{BurntOrange} $1$ & \color{BurntOrange}  $1$ & \color{BurntOrange}  $1$
& \color{BurntOrange} $1$ &
\color{BurntOrange} $0$ & \color{BurntOrange} $0$ & \color{BurntOrange} $0$ &
\color{BurntOrange} $0$ &
\color{BurntOrange} $0$ 
& \color{BurntOrange} $0$ & \color{BurntOrange}  $0$ & \color{BurntOrange}  $0$
& \color{BurntOrange} $0$ &
\color{BurntOrange} $0$ & \color{BurntOrange} $0$ & \color{BurntOrange} $0$ &
\color{BurntOrange} $0$ &
\color{BurntOrange} $0$ & \color{BurntOrange} &
        \multirow{2}{*}{ \Big\} \text{Cayley}}\\

& \color{BurntOrange} $0$ & \color{BurntOrange}  $0$ & \color{BurntOrange}  $0$
& \color{BurntOrange} $0$ &
\color{BurntOrange} $0$ & \color{BurntOrange} $0$ & \color{BurntOrange} $0$ &
\color{BurntOrange} $0$ &
\color{BurntOrange} $0$ 
& \color{BurntOrange} $0$ & \color{BurntOrange}  $0$ & \color{BurntOrange}  $0$
& \color{BurntOrange} $0$ &
\color{BurntOrange} $1$ & \color{BurntOrange} $1$ & \color{BurntOrange} $1$ &
\color{BurntOrange} $1$ &
\color{BurntOrange} $1$ 
& \color{BurntOrange} $1$ & \color{BurntOrange}  $1$ & \color{BurntOrange}  $1$
& \color{BurntOrange} $1$ &
\color{BurntOrange} $1$ & \color{BurntOrange} $1$ & \color{BurntOrange} $1$ &
\color{BurntOrange} $1$ &
\color{BurntOrange} $1$ & \color{BurntOrange} &\\

\cline{1-1}
\cline{29-29}

\end{tabular}
\] 

Implicitization requires eliminating the two parameters to obtain a
constraint equation over the symbolic coefficients of the polynomials.
Most of the coefficients are specialized except for $3$
variables, hence the sought for implicit equation of the surface 
is trivariate and the projection of $N(\R)$ lies in $\RR^3$.

TOPCOM~\citess{RambTOPCOM} 
needs more than a day and $9$GB of RAM to compute $1,806,467$
regular triangulations of $\A$, corresponding to $29$ of the vertices of
$N(\R)$, and crashes before computing the entire $N(\R)$.
Our algorithm yields the projected vertices
$\{(0,0,1),(0,1,0),$ $(1,0,0),(0,0,9),(0,18,0),(18,0,0)\}$
of the $3$-di\-men\-sional projection of $N(\R)$,
which is the Newton polytope of the implicit equation, in $30$msec.
Given this polytope, the implicit equation of the 
bicubic surface is interpolated in 42 seconds~\citess{EKKL12spm}. 
It is a polynomial of degree~$18$ containing $715$ terms which corresponds
exactly to the lattice points contained in the predicted polytope.
\end{example} 
 
Our main contribution is twofold.
First, we design an oracle-based algorithm for computing the
Newton polytope of $\mathcal{R}$, or of specializations of $\R$.  
The algorithm utilizes the Beneath-and-Beyond method to compute 
both vertex (V) and halfspace (H)
representations, which are required by the algorithm and may
also be relevant for the targeted applications.
Its incremental nature implies that we also obtain a triangulation of the
polytope, which may be useful for enumerating its lattice points.
The complexity is proportional to the number of output vertices and facets;
in this sense, the algorithms is output sensitive.
The overall cost is asymptotically dominated by computing as many
regular triangulations of $\A$ (Theorem~\ref{Ttotalcomplexity}).
We work in the space of the projected $N(\R)$ and revert to
the high-dimensional space of $\Sigma(\A)$ only if needed. 
Our algorithm readily extends to computing $\Sigma(\A)$, the Newton
polytope of the discriminant
and, more generally, any polytope that can
be efficiently described by a vertex oracle or its orthogonal projection.
In particular, it suffices to replace our oracle by the oracle in
Ref.~\refcite{Rincon12} to obtain a method for computing the discriminant
polytope.

Second, we describe an efficient, publicly available implementation
based on CGAL~\citess{CGAL}
and its experimental package {\tt triangulation}.
Our method computes instances of $5$-, $6$- or $7$-dimensional polytopes
with $35$K, $23$K or $500$ vertices, respectively, in $<2$hr.
Our code is faster up to dimensions $5$ or $6$, compared to a method computing 
$N(\R)$ via tropical geometry, implemented in the {\tt Gfan} 
library~\citess{JensenYu11}. In higher dimensions {\tt Gfan} seems to perform 
better although neither implementation can compute enough instances for a fair comparison.
Our code, in the critical step of computing the convex hull of the resultant
polytope, uses {\tt triangulation}.
On our instances, {\tt triangulation}, compared to state-of-the-art software
{\tt lrs}, {\tt cdd}, and {\tt polymake}, is the fastest together with {\tt
polymake}.
We factor out repeated computation by reducing the bulk of our
work to a sequence of determinants: this is often the case 
in high-dimensional geometric computing.
Here, we exploit the nature of our problem and matrix structure to capture the
similarities of the predicates, and hash the computed minors which are needed later,
to speedup subsequent determinants.
A variant of our algorithm computes successively tighter inner and outer
approximations: when these polytopes have, respectively,
90\% and 105\% of the true volume, runtime is reduced up to $25$ times.
This may lead to an approximation algorithm.

\paragraph{Previous work.}
Sparse (or toric) elimination theory was introduced in Ref.~\refcite{GKZ}.
They show that $N(\R)$,
for two univariate polynomials with $k_0+1,k_1+1$
monomials, has $\binom{k_0+k_1}{k_0}$ vertices and, when both $k_i\ge 2$,
it has $k_0k_1+3$ facets.
In Section~6~of~Ref.~\refcite{St94} is proven that $N(\R)$ is $1$-dimensional
if and only if $|A_i|=2$, for all $i$, the only planar $N(\R)$
is the triangle, whereas the only $3$-dimensional ones are the tetrahedron,
the square-based pyramid, and the resultant polytope of two univariate trinomials;
we compute an affinely isomorphic instance of the latter (Figure~\ref{fig:sec_res}(b))
as the resultant polytope of three bivariate polynomials.
Following Theorem~6.2~of~Ref.~\refcite{St94}, the $4$-dimensional polytopes
include
the 4-simplex, some $N(\R)$ obtained by pairs of univariate polynomials,
and those of~3 trinomials, which have been investigated with our
code in~Ref.~\refcite{DEF12}. 
The maximal (in terms of number of vertices) such polytope we have computed has
f-vector $(22,66,66,22)$ (Figure~\ref{fig:sec_res}(c)). 
Furthermore, Table~\ref{tbl:triang_size} presents some typical f-vectors of
$4,5,6$-dimensional projections of resultant polytopes. 

A lower bound on the volume of the Newton polytope of the discriminant polynomial 
that refutes a conjecture in algebraic geometry  
is presented in~Ref.~\refcite{discrim_vol}.

A direct approach for computing the vertices of $N(\R)$ might
consider all vertices of $\Sigma(\A)$ since the vertices of
the former are equivalence classes over the vertices of the latter.
Its complexity grows with the number of vertices
of $\Sigma(\A)$, hence is impractical (Example~\ref{ExamBicubic}).

The computation of secondary polytopes has been efficiently implemented in
TOPCOM~\citess{RambTOPCOM}, which has been the reference software for computing
regular or all triangulations.  
The software builds a search tree with flips as edges over the vertices of $\Sigma(A)$.  
This approach is limited by space usage.  
To address this, reverse search was proposed~\citess{IMTI02}, but the
implementation cannot compete with TOPCOM.
The approach based on computing $\Sigma(\A)$ is not efficient
for computing $N(\R)$.
For instance, in implicitizing parametric surfaces with up to $100$
terms, which includes all common instances in geometric modeling,
we compute the Newton polytope of the equations in less than $1$sec,
whereas $\Sigma(\A)$ is intractable (see e.g.\ Example~\ref{ExamBicubic}). 

In Ref.~\refcite{MicCoo00} they describe all Minkowski summands of $\Sigma(A)$.
In Ref.~\refcite{MicVer99} is defined an equivalence class over
$\Sigma(A)$ vertices having the same mixed cells.
The classes map in a many-to-one fashion to resultant vertices;
our algorithm exploits a stronger equivalence relationship.

Tropical geometry is a polyhedral analogue of algebraic geometry 
and can be viewed as generalizing sparse elimination theory.
It gives alternative ways of recovering resultant
polytopes~\citess{JensenYu11} and Newton polytopes of implicit
equations~\citess{StuYu08}. 
See Section~\ref{Simplement} for comparisons of 
the software in Ref.~\refcite{JensenYu11}, called {\tt Gfan},  
with our software.
In Ref.~\refcite{Rincon12}, tropical geometry
is used to define vertex oracles for the Newton polytope of the
discriminant polynomial.

In Ref.~\refcite{Hug06} there is a general implementation of a
Beneath-and-Beyond based procedure which reconstructs a polytope given by a
vertex  oracle. This implementation, as reported in~Ref.~\refcite{JensenYu11}, 
is outperformed by {\tt Gfan}, especially in dimensions higher than $5$. 

As is typical in computational geometry, the practical bottleneck
is in computing determinantal predicates.
For determinants, the record bit complexity is
\(O(n^{2.697})\)~\citess{KaVi05},
while more specialized methods exist for the sign of general determinants,
e.g. Ref.~\refcite{BEPP99}.
These results are relevant for higher dimensions and do not exploit the
structure of our determinantal predicates, nor the fact that we deal
with sequences of determinants whose matrices are not very different
(this is formalized and addressed in Section \ref{Shasheddets}).
We compared linear algebra libraries LinBox~\citess{DGGGHKSTV} and
Eigen~\citess{eigenweb}, which seem most suitable in dimension greater than $100$ and
medium to high dimensions, respectively, whereas CGAL provides the most
efficient determinant computation for the dimensions to which we focus.

The roadmap of the paper follows:
Section~\ref{Scombinatorics} describes the combinatorics of resultants, and
the following section presents our algorithm.
Section~\ref{Shasheddets} overcomes the bottleneck of Orientation predicates.
Section~\ref{Simplement} discusses the implementation, experiments, and 
comparison with other software.
We conclude with future work.  

A preliminary version containing most of the presented results 
appeared in Ref.~\refcite{EFKP12}. This extended version contains a more
detailed presentation of the background theory of resultants, applications
and examples, a more complete account of previous work, omitted proofs, 
an improved description of the approximation
algorithm, an extended version of the
hashing determinants method, and more experimental results.  

\section{Resultant polytopes and their projections}\label{Scombinatorics}

We introduce tools from combinatorial geometry~\citess{DeLRamSan,Ziegler}
to describe resultants~\citess{GKZ,CLO2}. 
We shall denote by vol$(\cdot)\in\RR$  
the normalized Euclidean volume,
$(\RR^m)^{\times}$ the linear $m$-dimensional functionals, $\Aff(\cdot)$ the affine hull, and 
$\CH(\cdot)$ the convex hull. 

Let $\A\subset\RR^d$ be a pointset whose convex hull is of dimension $d$. 
For any triangulation $T$ of $\A$, define vector $\phi_T\in\RR^{|\A|}$
with coordinate
\begin{equation}\label{Evolume_embed}
\phi_T(a)= \sum_{\sigma\in T : a\in \sigma} \mbox{vol}(\sigma), \qquad a\in\A,
\end{equation}
summing over all simplices $\sigma$ of $T$ having $a$ as a vertex; $\Sigma(\A)$
is the convex hull of $\phi_T$ for all triangulations $T$.
Let  $\A^w$ denote pointset $\A$ lifted to $\RR^{d+1}$
via a generic lifting function $w$ in $(\RR^{|\A|})^{\times}$.
{\em Regular triangulations} of $\A$ are obtained
by projecting the upper (or lower) hull of $\A^w$ 
back to $\RR^d$.
\begin{proposition}{\rm [Ref.~\refcite{GKZ}]}
The vertices of $\Sigma(\A)$ correspond to the regular triangulations of $\A$,
while its face lattice corresponds to the poset 
of regular polyhedral subdivisions of $\A$, ordered by refinement. 
A lifting vector produces a regular triangulation $T$
(resp.\ a regular polyhedral subdivision of $\A$)
if and only if it lies in the normal cone of vertex $\phi_T$
(resp.\ of the corresponding face) of $\Sigma(\A)$.  
The dimension of $\Sigma(\A)$ is $|\A|-d-1$. 
\end{proposition}
Let $A_0,\ldots,A_n$ be subsets of $\ZZ^n$,
$P_0,\ldots,P_n\subset \RR^n$ their convex hulls, and 
$P=P_0+\cdots+P_n$ their Minkowski sum.
A \emph{Minkowski (maximal) cell} of $P$ is any full-dimensional convex
polytope $B=\sum_{i=0}^{n} B_i$, where each $B_i$ is a convex polytope
with vertices in $A_i$.
Minkowski cells $B, B'=\sum_{i=0}^{n} B_i'$
intersect properly when $B_i\cap B_i'$ is a face of both and their
Minkowski sum descriptions are compatible, i.e.\ coincide on the common face.
A \textit{mixed subdivision} of $P$ is any family of
Minkowski cells which partition $P$ and intersect properly.
A Minkowski cell is \textit{$i$-mixed} or \textit{$v_{i}$-mixed},
if it is the Minkowski sum of $n$ one-dimensional segments from
$P_j,\, j\ne i$, and some vertex $v_{i} \in P_i$. In the sequel we shall call
a Minkowski cell, simply cell.

Mixed subdivisions contain {\em faces} of all dimensions between~0 and $n$,
the maximum dimension corresponding to cells.
Every face of a mixed subdivision of $P$ has a unique description as
Minkowski sum of $B_i\subset P_i$.
A mixed subdivision is {\em regular} if it is obtained as the projection
of the upper (or lower) hull of the Minkowski sum of lifted polytopes
$P_i^{w_i}:=\{(p_i,w_i(p_i))~|~ p_i \in P_i\}$, for lifting
$w_i:P_i\rightarrow\RR$.
If the lifting function $w:=(w_0\ldots,w_n)$ is sufficiently generic,
then the mixed subdivision is \emph{tight}, and
$\sum_{i=0}^n \dim B_i=\dim \sum_{i=0}^n B_i$, for every cell.
Given $A_0,\ldots,A_n$ and the affine basis $\{e_0,\ldots,e_n\}$ of $\RR^n$,
we define the Cayley pointset $\A\subset\ZZ^{2n}$ as in
equation~(\ref{EQ:Cayley}).

\begin{proposition}{\rm [Cayley
trick,~Ref.~\refcite{GKZ}]}\label{P:Cayley_trick}
There exist bijections between:
the regular tight mixed subdivisions of $P$ and
the regular triangulations of $\A$;
the tight mixed subdivisions of $P$ and the triangulations of $\A$;
the mixed subdivisions of $P$ and the polyhedral subdivisions of $\A$. 
\end{proposition} 

The family $A_0,\dots,A_n \subset\ZZ^n$ is \textit{essential}
if they jointly affinely span $\ZZ^n$ and every subset
of cardinality $j, 1\le j<n$, spans a space of dimension greater than or equal to $j$.
It is straightforward to check this property algorithmically and, if it
does not hold, to find an essential subset \citess{St94}. 
In the sequel, the input $A_0,\dots,A_n\subset \ZZ^n$ is supposed
to be essential. 
Given a finite $A\subset \ZZ^n$, we denote by $\CC^A$ the space of all Laurent polynomials
of the form $\sum_{a\in A} c_a x^a, c_a\neq 0, x=(x_1,\ldots,x_n)$.
Similarly, given $A_0,\dots,A_n \subset\ZZ^n$ we denote by $\prod_{i=0}^n \CC^{A_i}$
the space of all systems of polynomials
\begin{equation}\label{Esystem}
f_0=f_1=\dots=f_n=0,
\end{equation}
where $\sum_{a\in A_i} c_{i,a} x^a, c_{i,a}\neq 0$.
The vector of all coefficients $(\ldots,c_{i,a},\ldots)$ of \eqref{Esystem} defines 
a point in $\prod_{i=0}^n \CC^{A_i}$.
Let $Z\subset\prod_{i=0}^n \CC^{A_i}$ be the set of points corresponding to 
systems \eqref{Esystem} which have a solution in $(\CC^*)^n$,
and let $\overline{Z}$ be its closure. $\overline{Z}$ is an irreducible variety defined over $\QQ$.  

\begin{definition}\label{Dresultant}
If {\rm codim}$(\overline{Z})=1$, then the
\textit{sparse (or toric) resultant} of the system of polynomials \eqref{Esystem} 
is the unique (up to sign)
polynomial $\R$ in $\ZZ[ c_{i,a}:$ $i =0,\dots, n,$ $a\in A_i]$,
which vanishes on $\overline{Z}$.
If {\rm codim}$(\overline{Z})>2$, then $\R=1$.
\end{definition}

The resultant offers a solvability condition from which $x$
has been eliminated, hence is also known as the eliminant. 
For $n=1$, it is named after Sylvester.
For linear systems, it equals the determinant of the
$(n+1)\times (n+1)$ coefficient matrix.
The discriminant of a polynomial $F(x_1,\dots,x_n)$ is given by the resultant
of $F,\partial F / \partial x_1, \dots,$ $\partial F / \partial x_n$.

The Newton polytope $N(\R)$ of the resultant is a lattice polytope called the 
\textit{resultant polytope}. The resultant
has $|\A|=\sum_{i=0}^n |A_i|$ variables, hence $N(\R)$ lies in $\RR^{|\A|}$,
though it is of smaller dimension (Proposition~\ref{Psummand_dimRes}).
The monomials corresponding to vertices of $N(\R)$ are the
extreme resultant monomials.
\begin{proposition}{\rm [Refs.~\refcite{GKZ,St94}]} \label{PSturmf_extreme}
For a sufficiently generic lifting function $w\in (\RR^{|\A|})^{\times}$, 
the $w$-extreme monomial of $\R$, whose exponent vector 
maximizes the inner product with $w$, equals
\begin{equation}\label{Eq:Sturmf_extreme}
\pm \prod_{i=0}^{n} \prod_{\sigma} c_{i,v_{i}}^{\mathrm{vol}(\sigma)},
\end{equation}
where $\sigma$ ranges over all $v_i$-mixed cells of the regular tight mixed
subdivision $S$ of $P$ induced by $w$, and $c_{i,v_i}$ is the
coefficient of the monomial $x^{v_i}$ in $f_i$.
\end{proposition}
Let $T$ be the regular triangulation corresponding, via the Cayley trick,
to $S$, and $\rho_T\in\NN^{|\A|}$ the exponent of the $w$-extreme monomial.
For simplicity we shall denote by $\sigma$, both a cell of $S$  and its 
corresponding simplex in $T$. Then,
\begin{equation}\label{Eq:rho}
\rho_T(a)= \sum_{ \stackrel{a-\text{mixed}}{\sigma\in T : a\in \sigma}}
\mbox{vol}(\sigma)\; \in\NN, \qquad a\in\A ,
\end{equation}
where simplex $\sigma$ is $a$-mixed if and only if the corresponding cell 
is $a$-mixed in $S$. 
Note that, $\rho_T(a) \in \NN$, since it is
a sum of volumes of mixed simplices $\sigma\in T$, and each of these volumes is
equal to the \emph{mixed volume}\citess{CLO2} of a set of \emph{lattice}
polytopes, the Minkowksi summands of the corresponding $\sigma\in S$. In
particular, assuming that $\sigma \in S$ is $i$-mixed, it can be written as
$\sigma=\sigma_0+\dots+\sigma_n,\; \sigma_j \subseteq A_j,\, j=0,\ldots,n$, and
$
 \mbox{vol}(\sigma)=MV(\sigma_0,\ldots,\sigma_{i-1},\sigma_{i+1},\ldots,\sigma_n),
$
where $MV$ denotes the mixed volume function which is integer valued for 
lattice polytopes~\citess{CLO2}.
Now, $N(\R)$ is the convex hull of all $\rho_T$ vectors~\citess{GKZ,St94}.

\begin{figure*}[t] \centering
 \includegraphics[width=\textwidth]{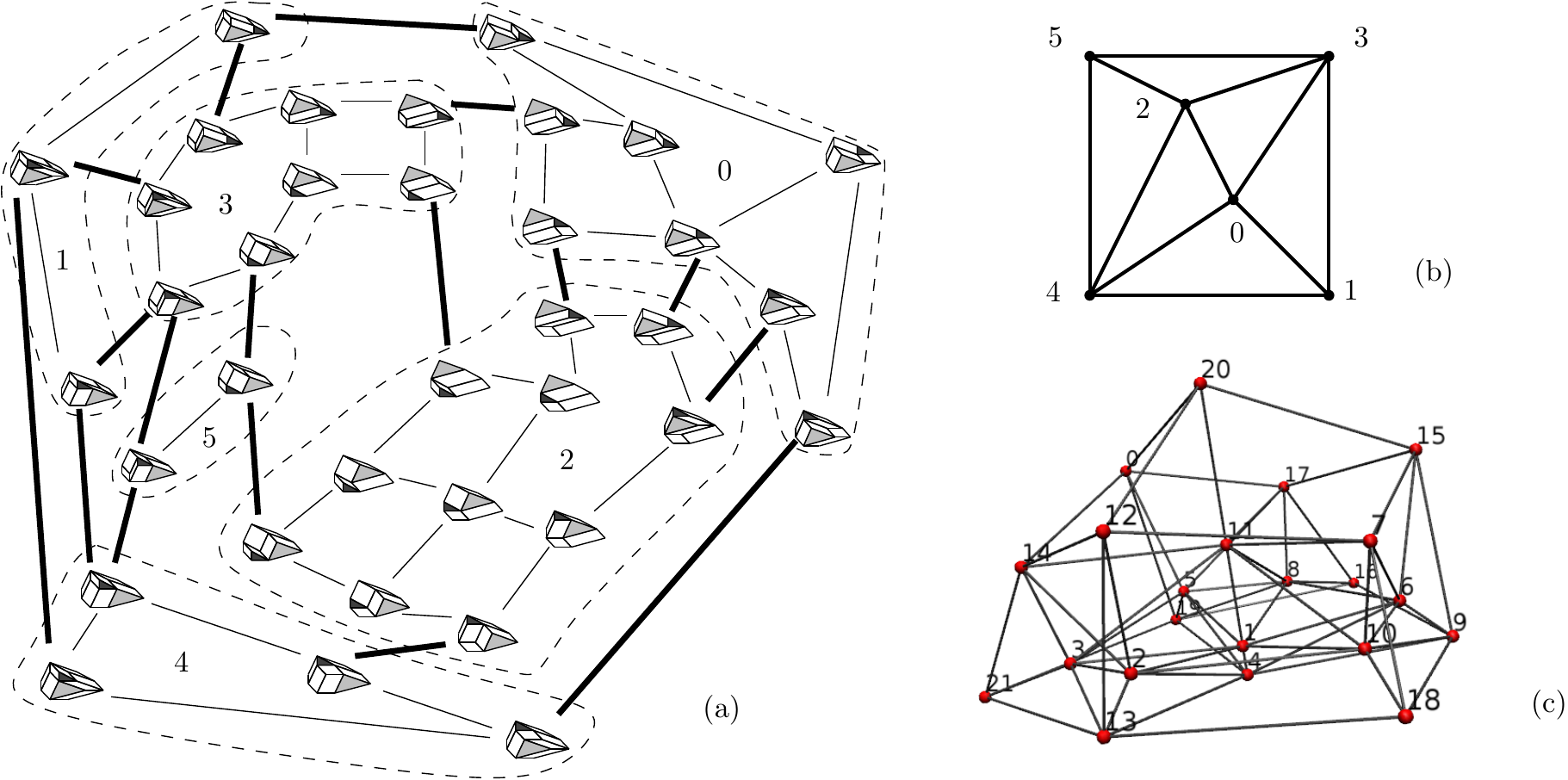}
 \caption{(a) 
The secondary polytope $\Sigma(\A)$ of two triangles (dark, light grey) and one
segment
$A_0=\{(0,0),(1,2),(4,1)\},\,A_1=\{(0,1),(1,0)\},\,A_2=\{(0,0),(0,1),(2,0)\}$,
where $\A$ is defined as in Equation~\ref{EQ:Cayley};
vertices correspond to mixed subdivisions of the
Minkowski sum $A_0+A_1+A_2$ and edges to flips between them
(b)
$N(\R)$, whose vertices correspond to the dashed classes of $\Sigma(\A)$.
Bold edges of $\Sigma(\A)$, called cubical flips, map to edges of $N(\R)$
(c)
$4$-dimensional $N(\R)$ of 3 generic trinomials with f-vector $(22,66,66,22)$;
figure made with {\tt polymake}.
}
\label{fig:sec_res} \end{figure*}

Proposition~\ref{PSturmf_extreme} establishes
a many-to-one surjection from regular triangulations of $\A$ to
regular tight mixed subdivisions of $P$, or, equivalently,
from vertices of $\Sigma(\A)$ to those of $N(\R)$.
One defines an {\em equivalence relationship} on all regular tight mixed
subdivisions, where equivalent subdivisions yield the same vertex in $N(\R)$. 
Thus, equivalent vertices of $\Sigma(\A)$ correspond to the same
resultant vertex.
Consider $w\in (\RR^{|\A|})^{\times}$ lying in the union of
outer-normal cones of equivalent vertices of $\Sigma(\A)$.
They correspond to a resultant vertex whose outer-normal cone
contains $w$; this defines a $w$-extremal resultant monomial.
If $w$ is non-generic, it specifies a sum of extremal monomials in $\R$,
i.e. a face of $N(\R)$.
The above discussion is illustrated in Figure~\ref{fig:sec_res}(a),(b).

\begin{proposition}\label{Psummand_dimRes}{\rm [Ref.~\refcite{GKZ}]}
$N(\R)$ is a Minkowski summand of 
$\Sigma(\A)$, and both $\Sigma(\A)$ and $N(\R)$ have 
dimension $|\A| -2n-1.$
\end{proposition} 

Let us describe the $2n+1$ hyperplanes in whose intersection lies $N(\R)$.
For this, let $M$ be the $(2n+1)\times |\A|$ matrix whose columns are the
points in the $A_i$, where each $a\in A_i$ is followed by
the $i$-th unit vector in $\NN^{n+1}$.
Then, the inner product of any coordinate vector of $N(\R)$ with row $i$
of $M$ is: constant, for $i=1,\dots,n$, and known, and depends on $i$,
for $i=n+1,\dots,2n+1$, see Prop.~7.1.11 of Ref.~\refcite{GKZ}.
This implies that one obtains an isomorphic polytope when projecting
$N(\R)$ along $2n+1$ points in $\A$ which affinely span $\RR^{2n}$;
this is possible because of the assumption of essential family.
Having computed the projection, we obtain $N(\R)$ by computing
the missing coordinates as the solution of a linear system:
we write the aforementioned inner products as $M [ X\, V]^T = C$, where $C$
is a known matrix and $[ X\, V]^T$ is a transposed $(2n+1)\times u$ matrix,
expressing the partition of the coordinates to unknown and known values,
where $u$ is the number of $N(\R)$ vertices. 
If the first $2n+1$ columns of $M$ correspond to specialized coefficients, 
$M = [M_1\, M_2]$, where submatrix $M_1$ is of dimension $2n+1$
and invertible, hence $X=M_1^{-1}(C-M_2B)$.

We compute some orthogonal projection of $N(\R)$, denoted
$\varPi$, in $\RR^m$:
$$
\pi:\RR^{|\A|}\rightarrow \RR^m : N(\R)\rightarrow \varPi,\;\, m \le |\A| .
$$
By reindexing, this is the subspace of the first $m$ coordinates,
so\linebreak
$\pi(\rho)=(\rho_1,\dots,\rho_m)$.
It is possible that none of the coefficients $c_{ij}$ is specialized,
hence $m = |\A|$, $\pi$ is trivial, and $\varPi=N(\R)$.
Assuming the specialized coefficients take sufficiently
generic values, $\varPi$ is the Newton polytope of the corresponding
specialization of $\R$.
The following is used for preprocessing.

\begin{lemma}{\rm [Ref.~\refcite{JensenYu11}~Lemma~3.20]}\label{Linsidepoints}
If $a_{ij} \in A_i$ corresponds to a specialized coefficient of $f_i$,
and lies in the convex hull of the other points in $A_i$ corresponding
to specialized coefficients, then removing $a_{ij}$ from $A_i$ does not
change the Newton polytope of the specialized resultant.
\end{lemma} 

We focus on three applications. First,
we interpolate the resultant in all coefficients, thus illustrating
an alternative method for computing resultants.
\begin{example}\label{ExamGenRes}
Let $f_0=a_2x^2+a_1x+a_0$, $f_1=b_1x^2+b_0$, with supports
$A_0=\{2,1,0\}, A_1=\{1,0\}$.
Their (Sylvester) resultant is a polynomial in $a_2, a_1, a_0,$ $b_1, b_0$.
Our algorithm computes its Newton polytope with vertices
$(0, 2, 0, 1, 1)$,  $(0, 0, 2, 2, 0)$, $(2, 0, 0, 0, 2)$;
it contains 4 lattice points, corresponding to 4 potential resultant monomials
$a_1^2 b_1  b_0,~ a_0^2 b_1^2,~  a_2 a_0 b_1 b_0,~ a_2^2 b_0^2$.
Knowing these potential monomials, to interpolate the resultant, we need 4 
points $(a_0, a_1, a_2, b_0, b_1)$
for which the system $f_0= f_1=0$ has a solution.
For computing these points we use the
parameterization 
of resultants in Ref.~\refcite{Kap91}, which yields:
$a_2=(2t_1+t_2)t_3^2t_4$,  $a_1=(-2t_1-2t_2)t_3 t_4$,
$a_0=t_2t_4$, $b_1=-t_1t_3^2t_5$, $b_0=t_1t_5,$
where the $t_i$'s are parameters.
We substitute these expressions to the monomials,
evaluate at~4 sufficiently random $t_i$'s, and obtain a matrix
whose kernel vector $(1, 1, -2, 1)$ yields
$\R= a_1^2 b_1 b_0 + a_0^2 b_1^2 -2a_2 a_0 b_1 b_0+a_2^2b_0^2$.
\end{example}

Second, consider system solving by the
rational univariate representation of roots~\citess{BaPoRo}. 
Given $f_1,\dots,f_n\in \CC[x_1,\dots,x_n]$,
define an overconstrained system by adding
$f_0=u_0 + u_1x_1+\cdots + u_nx_n$ with symbolic $u_i$'s.
Let coefficients $c_{ij}, i\ge 1$, take specific values,
and suppose that the roots of $f_1 = \cdots =f_n=0$ are isolated, denoted
$r_i=(r_{i1},\dots,r_{in})$.
Then the $u$-resultant is
$\R_u= a\, \prod_{r_i} (u_0+u_1r_{i1}+\cdots + u_nr_{in})^{m_i}$, 
$a\in\CC^*$, where $m_i$ is the multiplicity of $r_i$.
Computing $\R_u$ is the bottleneck; our method computes 
(a superset of) $N(\R_u)$.

\begin{example}\label{ExamUres}
Let $f_1=x_1^2+x_2^2-4$, $f_2=x_1-x_2+2$, and $f_0=u_0+u_1x_1+u_2x_2$.
Our algorithm computes a polygon with vertices 
$\{(2,0,0),(0,2,0),(0,0,2)\}$, which contains $N(\R_u)=
\CH(\{(2, 0, 0), (1, 1, 0), (1, 0, 1), (0,1,1)\})$. 
The coefficient specialization is not generic, hence $N(\R_u)$ is
strictly contained in the computed polygon.
Proceeding as in Example~\ref{ExamGenRes},
$\R_u= 2u_0^2+4u_0u_1-4u_0u_2-8u_1u_2$, which factors as
$2 (u_0 + 2u_1) (u_0-2u_2) $.
\end{example} 

The last application comes from geometric modeling, where
$y_i=f_i(x)$, $i=0,\dots,n$, $x=(x_1,\dots,x_n)\in\Omega\subset\RR^n$,
defines a parametric hypersurface. 
Many applications 
require the equivalent implicit representation
$F(y_1,\dots,y_n)=0$.
This amounts to eliminating $x$, so 
it is crucial to compute the resultant 
when coefficients are specialized except the $y_i$'s.
Our approach computes a polytope
that contains the Newton polytope of $F$, thus reducing implicitization
to interpolation~\citess{EKKL12spm,EmKaKoLB}.
In particular, we compute the polytope of surface equations within $1$sec,
assuming $\le 100$ terms in parametric polynomials,
which includes all common instances in geometric modeling.

\begin{figure*}[t]\centering
\raisebox{6.5mm}{
\includegraphics[width=0.6\textwidth]{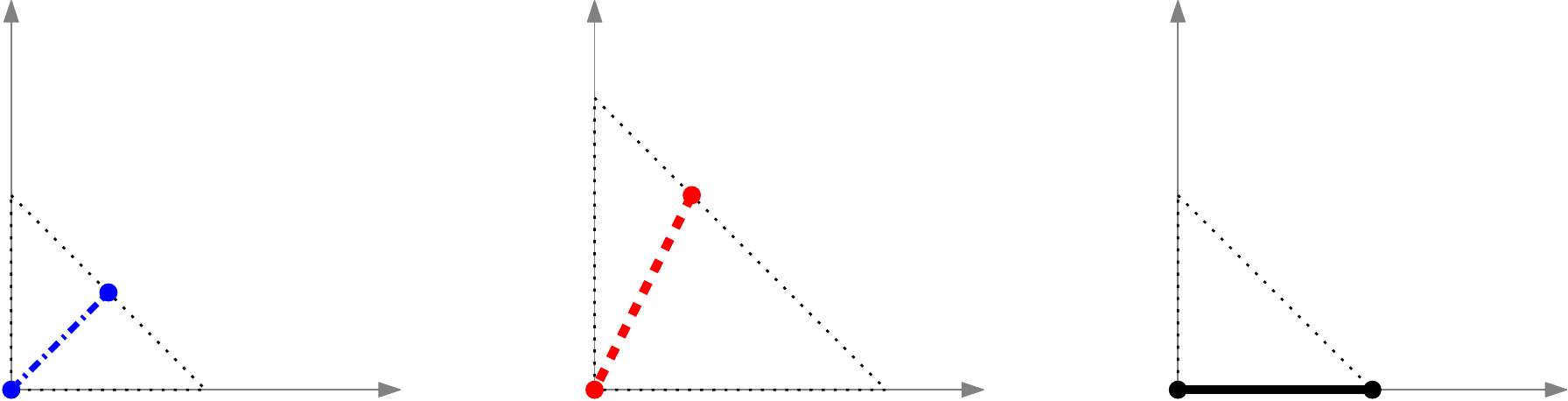}}
\qquad \includegraphics[width=0.6\textwidth]{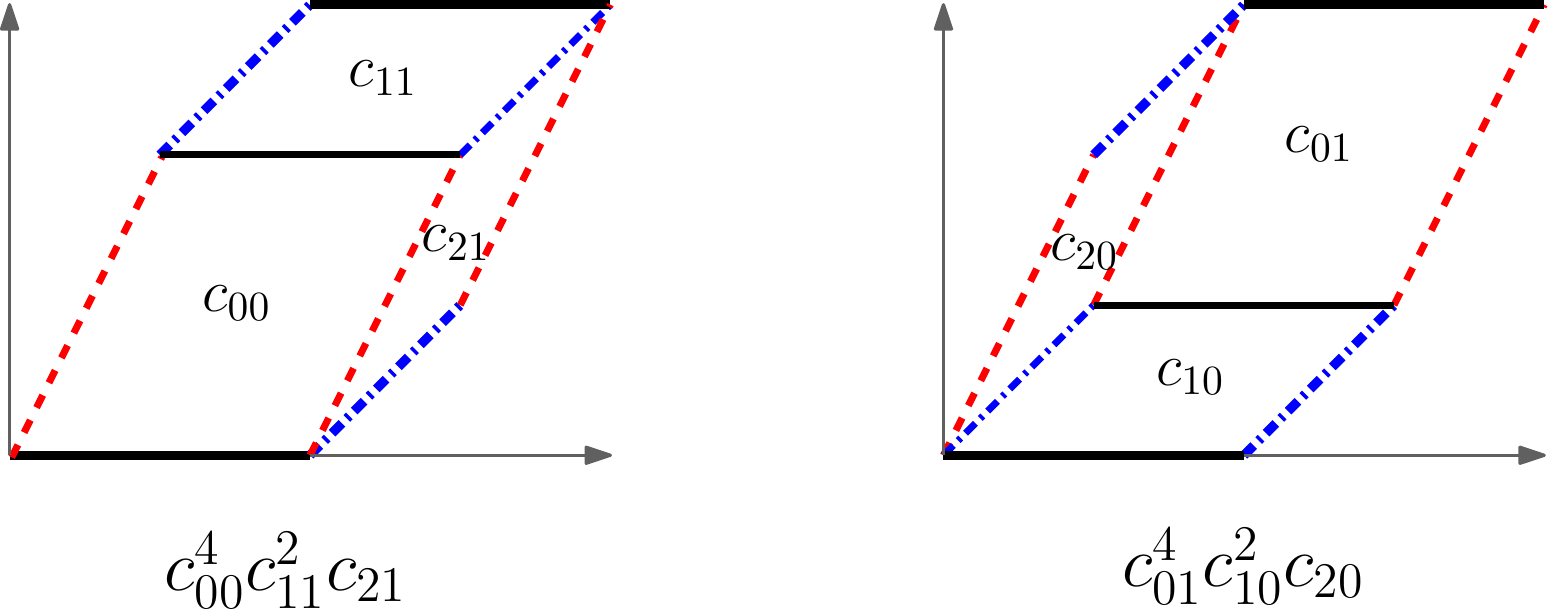} 
\caption[]{The supports $A_0, A_1, A_2$ of Example~\ref{ExamBuchberger}, their
Newton polytopes (segments) and the two mixed subdivisions of their Minkowski
sum.
\label{Fbuchberger}} 
\end{figure*}

\begin{example} \label{ExamBuchberger} 
Let us see how the above computation can serve in implicitization.
Consider the surface given by the polynomial parameterization 
$$ (y_1,y_2,y_3) = (x_1 x_2, x_1 x_2^2, x_1^2).  $$
For polynomials
$f_0:=c_{00}-c_{01}x_1x_2,~ f_1:=c_{10}-c_{11}x_1x_2^2,~f_2:=c_{20}-c_{21}x_1^2$
with supports
$A_0=\{(0,0), (1,2)\}, A_1=\{(0,0), (1,2)\}$ and $A_2=\{(0,0), (2,0)\}$.
The resultant polytope is a segment in $\RR^6$ with endpoints
$(4,0,0,2,0,1)$,  $(0,4,2,0,1,0)$ and, actually,
$\R=-c_{00}^4c_{11}^2c_{21}+c_{01}^4c_{10}^2c_{20}$.
The supports and the two mixed subdivisions corresponding to the vertices of
$N(\R)$ are illustrated in Figure~\ref{Fbuchberger}.
Specializing the symbolic coefficients of the polynomials as: 
$$ (c_{00}, c_{01}, c_{10}, c_{11}, c_{20}, c_{21})\mapsto (y_1, -1, y_2, -1,
y_3, -1) $$
yields the vertices of the implicit polytope: $(4,0,0), (0,2,1)$,
which our algorithm can compute directly.
The implicit equation of the surface turns out to be $-y_1^4 + y_2^2 y_3$.
\end{example}

\section{Algorithms and complexity}\label{Sproject}
%
 \newcommand{\Li}{{\cal L}}
 \newcommand{\Sub}{{\cal S}}
 \newcommand{\T}{{T}}
 \newcommand{\J}{{\cal J}}
 \newcommand{\illH}{{\cal H}_{illegal}}
 \newcommand{\W}{{\cal W}}
 \newcommand{\Q}{Q^{H}}
 \newcommand{\Qo}{Q_o^{H}}
 \newcommand{\V}{Q}

This section analyzes our exact and approximate algorithms for computing
orthogonal projections of polytopes whose vertices are defined by an
\emph{oracle}.
This oracle computes a 
vertex of the polytope which is extremal in a given
direction $w$. If there are more than one such vertices 
the oracle returns exactly one of these. Moreover, we define such an oracle for the
vertices of orthogonal projections $\varPi$ of $N(\R)$ which results in algorithms for
computing $\varPi$ while avoiding computing $N(\R)$. 
Finally, we analyze the
asymptotic complexity of these algorithms.

Given a pointset $V$, {reg\_subdivision($V, \omega$)} 
computes the regular subdivision of its convex hull by projecting 
the upper hull of $V$ lifted by $\omega$, and
{conv($V$)} computes the H-representation of the convex hull of $V$.
The oracle {VTX}($\A,\,w,\,\pi$) 
computes a point in $\varPi=\pi(N(\R))$, extremal
in the direction $w\in(\RR^m)^\times$. 
First it adds to $w$ an infinitesimal symbolic
perturbation vector, thus obtaining $w_p$. 
Then calls reg\_subdivision($\A,
\widehat{w_p}$), $\widehat{w_p}=(w_p,\vec{0})\in (\RR^{|\A|})^\times$ that yields 
a regular triangulation $T$ of $\A$, since $w_p$ is generic,
and finally returns $\pi(\rho_T)$. 
It is clear that the 
triangulation $T$ constructed by {VTX}$(\cdot)$ is regular and
corresponds to some secondary vertex $\phi_T$ which maximizes the
inner product with $\widehat{w_p}$. 
Since the perturbation is arbitrarily small, both $\phi_T, \rho_T$ also maximize the
inner product with $\widehat{w}=(w,\vec{0})\in (\RR^{|\A|})^\times$.

We use perturbation to avoid computing non-vertex points on the boundary of $\varPi$.
The perturbation can be implemented in VTX$(\cdot)$, without
affecting any other parts of the algorithm, either by case analysis or by
a method of symbolic perturbation.
In practice, our implementation does avoid computing non-vertex points on 
the boundary of $\varPi$ by 
computing a refinement of the subdivision obtained by calling 
reg\_subdivision($\A, \widehat{w}$). 
This refinement is implemented in {\tt triangulation} 
by computing a placing triangulation 
with a random insertion order~\citess{BoiDevHor09} (Section~\ref{Simplement}).

\begin{lemma}\label{Lpointonboundary}
All points computed by {\em VTX}$(\cdot)$ are vertices of $\varPi$.
\end{lemma} 
\begin{proof}
Let $v=\pi(\rho_T)=\mbox{VTX}(\A,w,\pi)$.
We first prove that $v$ lies on $\partial\varPi$.
The point $\rho_T$ of $N(\R)$ is a Minkowski summand of the vertex $\phi_T$ of
$\Sigma(\A)$ extremal with respect to $\widehat{w}$, hence $\rho_T$ is extremal
with respect to $\widehat{w}$. 
Since $\widehat{w}$ is perpendicular to projection $\pi$, $\rho_T$
projects to a point in $\partial\varPi$.
The same argument implies that every vertex $\phi_T'$, where $T'$ is a
triangulation refining 
the subdivision produced by $\widehat{w}$, corresponds to a  resultant vertex
$\rho_{T'}$ such that $\pi(\rho_{T'})$ lies on a face
of $\varPi$. This is actually the same face on which $\pi(\rho_T)$ lies.
Hence $\rho_{T'}$ also lies on $\partial\varPi$.

Now we prove that $v$ is a vertex of $\varPi$ 
by showing that it does not lie in the 
relative interior of a face of $\Pi$.
Let $w$ be such that 
the face $f$ of $N(\R)$ extremal with respect to
$\widehat{w}$ contains a vertex $\rho_T$ which projects to
$\mbox{relint}(\pi(f))$, where $\mbox{relint}(\cdot)$ denotes relative interior.
However, $f$ will not be extremal with respect 
to $\widehat{w_p}$ and since VTX$(\A,w,\pi)$ uses the perturbed vector $w_p$, 
it will never compute a vertex of  $N(\R)$ whose projection lies inside a face 
of $\varPi$. 
\end{proof}

The \textit{initialization algorithm}
computes an inner approximation of $\varPi$
in both V- and H-representations (denoted $\V,\ \Q$, respectively),
and triangulated.
First, it calls {VTX}$(\A,w,\pi)$ for
$w\in W\subset(\RR^m)^\times$; the set $W$ is either random or contains,
say, vectors in the $2m$ coordinate directions.
Then, it updates $Q$ by adding {VTX}$(\A,w,\pi)$ and {VTX}$(\A,-w,\pi)$,
where $w$ is normal to hyperplane $H\subset \RR^m$ containing $Q$,
as long as either of these points lies outside $H$.
Since every new vertex lies outside the affine hull of the current 
polytope $Q$, all polytopes produced are simplices.
We stop when these points do no longer increase $\dim(Q)$. 

\begin{lemma}\label{Linit}
The initialization algorithm computes $Q\subseteq \varPi$
such that $\dim(Q)=\dim(\varPi)$.
\end{lemma}
\begin{proof}
Suppose that the initialization algorithm computes a polytope $Q'\subset \varPi$ 
such that $\dim(Q')<m$. Then there exists vertex
$v \in \varPi$, $v\notin \Aff(Q')$ and vector
$w\in (\RR^m)^\times$ perpendicular to $\Aff(Q')$,
such that $w$ belongs to the normal cone of $v$
in $\varPi$ and $\dim(\Aff(Q' \cup v)) > \dim Q'$.
This is a contradiction, since such a $w$ would have been computed
as VTX($\A,w,\pi$) or VTX($\A,-w,\pi$),
where $w$ is normal to the hyperplane $H$ containing $Q'$.
\end{proof} 

Incremental Algorithm~\ref{AlgComputeP} computes both V- and
H-representa\-tions of $\varPi$ and a triangulation
of $\varPi$, given an inner approximation $\V,\Q$ of $\varPi$ computed at
the initialization.
A hyperplane $H$ is called \emph{legal}
if it is a supporting hyperplane to a facet of $\varPi$,
otherwise it is called \emph{illegal}.
At every step of Algorithm~\ref{AlgComputeP}, we compute
$v=\mbox{VTX}(\A,w,\pi)$ for a supporting hyperplane $H$ of a facet of
$Q$
with normal $w$.
If $v\notin H$, it is a new vertex thus yielding a tighter \textit{inner
approximation} of $\varPi$ by inserting it to $Q$, i.e.\ $Q\subset
\CH(Q\cup v) \subseteq \varPi$. 
This happens when the preimage  $\pi^{-1}(f)\subset N(\R)$ of the facet $f$ of
$Q$ 
defined by $H$, is not a Minkowski summand of a face of $\Sigma(\A)$ having 
normal $\widehat{w}$.
Otherwise, there are two cases: either $v\in H$ and $v\in Q$, 
thus the algorithm simply decides hyperplane $H$ is legal, or
$v\in H$ and $v \notin Q$, in which case the algorithm again
decides $H$ is legal but also inserts $v$ to $Q$.

The algorithm computes $\Q$ from $\V$, then iterates over the
new hyperplanes to either compute new vertices or decide they are legal,
until no increment is possible, which happens when all hyperplanes
are legal. 
Algorithm~\ref{AlgComputeP} ensures that each normal $w$ to a hyperplane
supporting a facet of $Q$ is used only \emph{once}, by storing all used $w$'s in
a set $W$.
When a new normal $w$ is created, the algorithm checks if $w\notin W$,
then calls VTX$(\A,w,\pi)$ and updates $W\leftarrow W\cup w$.
If $w\in W$ then the same or a parallel hyperplane has been
checked in a previous step of the algorithm.
It is straightforward that $w$ can be safely ignored;
Lemma~\ref{Noparallel} formalizes the latter case.

\begin{lemma}\label{Noparallel}
Let $H'$ be a hyperplane supporting a facet constructed
by Algorithm~\ref{AlgComputeP}, and $H\ne H'$ an illegal hyperplane
at a previous step.
If $H',H$ are parallel then $H'$ is legal.
\end{lemma}
\begin{proof}
Let $w,w'$ be the outer normal vectors of the 
facets supported by $H,H'$ respectively.
If $H,H'$ are parallel then $v=\mbox{VTX}(\A,w,\pi)$ maximizes the 
inner product with $w'$ 
in $Q$
which implies that hyperplane $H'$ is legal.
\end{proof}
\begin{algorithm}[ht]
  \BlankLine
  \Input{essential $A_0,\dots,A_n\subset\ZZ^n$ processed by
	Lemma~\ref{Linsidepoints},\\
         projection $\pi:\RR^{|\A|}\rightarrow \RR^m$,\\
H-, V-repres.~$\Q,\V$; triang.~$\T_{Q}$ of $Q\subseteq\varPi$.}
  \Output{H-, V-repres.~$\Q,\V$; triang.~$\T_Q$ of $Q=\varPi$.}
  \BlankLine\BlankLine
  $\A\leftarrow \bigcup_{0}^{n}(A_i\times{e_i})$
  	\hspace{1em}\tcp{Cayley trick}
  $\illH\leftarrow\emptyset$\; 
  \lForEach{$H \in \Q$}{$\illH\leftarrow \illH\cup\{H\}$\hspace{2em}}
  \BlankLine
  \While{$\illH\neq\emptyset$}{
    select $H \in \illH$ and $\illH\leftarrow\illH\setminus\{H\}$\;
    $w$ is the outer normal vector of $H$\;
      $v\leftarrow$ VTX($\A,w,\pi$)\;
      \If {$v\notin H\cap \V$}{
	$\Q_{temp} \leftarrow {\rm conv}(\V \cup \{v\})$ \tcp{convex hull computation}
        \ForEach{$(d-1)$-face $f\in\T_Q$ visible from $v$}{
          $\T_Q\leftarrow\T_Q\cup\{\text{faces of }{conv}(f,v)\}$
        } 
        \ForEach{$H' \in \{\Q\setminus \Q_{temp}\}$}{
	  $\illH\leftarrow \illH\setminus\{H'\}$ \tcp{$H'$ separates $\V,v$}
	}
        \ForEach{$H' \in \{\Q_{temp}\setminus \Q\}$}{
	  $\illH\leftarrow \illH\cup\{H'\}$ \tcp{new hyperplane}
	}
        $\V \leftarrow \V \cup \{v\}$\;
        $\Q\leftarrow \Q_{temp}$\;
      }
  }
  \Return $\V,\Q,\T_Q$\;
  \BlankLine
  \caption{\label{AlgComputeP} Compute$\varPi$ $(A_0,\dots,A_n,\pi)$}
\end{algorithm} 

The next lemma formulates the termination criterion of our algorithm.

\begin{lemma}\label{Lwcriterion} 
Let $v=\mbox{{\em VTX}}(\A,w,\pi)$, where $w$ is
normal to a supporting hyperplane $H$ of $Q$,
then $v\not\in H$ if and only if $H$ is not a supporting hyperplane of $\varPi$.
\end{lemma}

\begin{proof}
Let $v=\pi(\rho_T)$, where $T$ is a
triangulation refining subdivision $S$ in {VTX}$(\cdot)$.
It is clear that, since $v\in \partial\varPi$ is extremal with respect to $w$,
if $v\not\in H$ then $H$ cannot be a supporting hyperplane of $\varPi$.
Conversely, let $v\in H$.
By the proof of Lemma~\ref{Lpointonboundary}, every other vertex
$\pi(\rho{_T'})$ 
on the face of $N(\R)$ is extremal with respect to $w$, hence lies on $H$,
thus $H$ is a supporting hyperplane of $\varPi$.
\end{proof}

We now bound the {\it complexity} of our algorithm.
Beneath-and-Beyond, given a $k$-dimensional polytope with
$l$ vertices, computes its H-representation and a triangulation
in $O(k^5lt^2)$, where $t$ is the number of full-dimensional faces (cells)
Ref.~\refcite{Josw03bb}.
Let $|\varPi|, |\varPi^H|$ be the number of vertices and facets of $\varPi$.

\begin{lemma}\label{Loneperhplane}
Algorithm~\ref{AlgComputeP} 
executes {VTX}$(\cdot)$
at most $|\varPi| +|\varPi^H|$ 
times.
\end{lemma}
\begin{proof}
The steps of Algorithm~\ref{AlgComputeP} increment $Q$.
At every such step, and for each supporting hyperplane $H$ of $Q$ with normal
$w$, 
the algorithm calls {VTX}$(\cdot)$
and computes one vertex of $\varPi$, 
by Lemma~\ref{Lpointonboundary}.
If $H$ is illegal, this vertex is unique 
because $H$ separates the set of (already computed) vertices of $Q$ from the set
of vertices of $\varPi\setminus Q$ which are extremal with respect to $w$,
hence, an appropriate translate of $H$ also separates the corresponding sets of
vertices of $\Sigma(\A)$ (Figure~\ref{fig:oneperhplane}).
This vertex is never computed again because it now belongs to $Q$.
The number of {VTX}$(\cdot)$ calls
yielding vertices is thus bounded by $|\varPi|$.

For a legal hyperplane of $Q$, we compute one vertex of $\varPi$
that confirms its legality; the 
{VTX}$(\cdot)$ call
yielding this vertex is accounted for by the legal hyperplane.
The statement follows by observing that 
every normal to a hyperplane of $Q$ is used only once
in Algorithm~\ref{AlgComputeP}
(by the earlier discussion concerning the set $W$ of all used normals).
\end{proof}

\begin{figure}[t]
\centering
\includegraphics[scale=0.9]{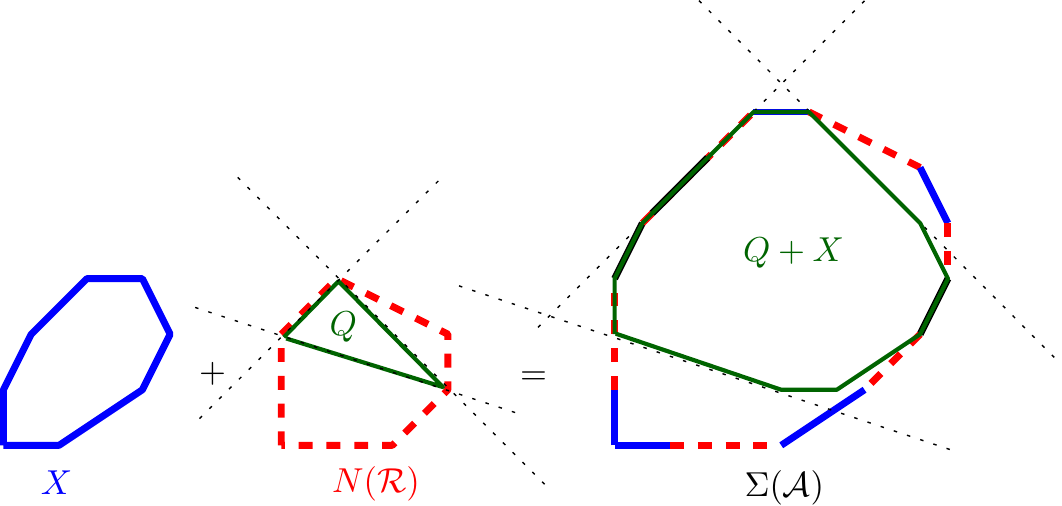}
\caption{Lemma~\ref{Loneperhplane}: each illegal hyperplane of $Q$ with normal
$w$, separates
the already computed vertices of $\varPi$ (here equal to $N(\R)$) from new ones,
extremal with respect to $w$. $X$ is a polytope
such that $X+N(\R)=\Sigma(\A)$.}
\label{fig:oneperhplane}
\end{figure}

Let the size of a triangulation be the number of its cells.
Let $s_{\A}$ denote the size of the largest triangulation of $\A$
computed by {VTX}$(\cdot)$, and $s_{\varPi}$ that
of $\varPi$ computed by Algorithm~\ref{AlgComputeP}.
In {VTX}$(\cdot)$, the computation of a regular triangulation reduces to a
convex hull, computed in $O(n^5|\A|s_{\A}^2)$;
for $\rho_{\T}$ we compute Volume for all cells of $\T$ in $O(s_{\A} n^3)$.
The overall complexity of {VTX}$(\cdot)$ becomes $O(n^5|\A|s_{\A}^2)$.
Algorithm~\ref{AlgComputeP} calls, in every step, {VTX}$(\cdot)$ 
to find a point on $\partial \varPi$ and insert it to $\V$,
or to conclude that a hyperplane is legal.
By Lemma~\ref{Loneperhplane} it executes {VTX}$(\cdot)$ as many as 
$|\varPi|+|\varPi^H|$ 
times, in 
$O((|\varPi|+|\varPi^H|)n^5|\A| s_{\A}^2)$, and
computes the H-representation of $\varPi$ in $O(m^5|\varPi| s_{\varPi}^2)$.
Now we have,
$|\A|\leq(2n+1)s_{\A}$ and as the input $|\A|,m,n$ grows large
we can assume that $|\varPi|\gg |\A|$ and thus $s_{\varPi}$ dominates
$s_{\A}$. Moreover, $s_{\varPi}(m+1)\geq|\varPi^H|$.
Now, let
$\widetilde{O}(\cdot)$ imply that polylogarithmic factors are ignored.

\begin{theorem}\label{Ttotalcomplexity}
The time complexity of Algorithm~\ref{AlgComputeP} to compute $\varPi\subset\RR^m$
is $O(m^5|\varPi| s_{\varPi}^2+(|\varPi|+|\varPi^H|)n^5|\A| s_{\A}^2)$,
which becomes $\widetilde{O}(|\varPi| s_{\varPi}^2)$
when $|\varPi|\gg |\A|$.
\end{theorem}

This implies our algorithm is output sensitive.
Its experimental performance confirms this property, see Section~\ref{Simplement}.

We have proven that oracle {VTX}$(\cdot)$ (within our algorithm)
has two important properties:\\[-14pt]
\begin{enumerate}\itemsep0pt
 \item Its output is a vertex of the target polytope
(Lemma~\ref{Lpointonboundary})\label{oracle1}.
\item When the direction $w$ is normal to an illegal facet,
then the vertex computed by the oracle is computed once
(Lemma~\ref{Loneperhplane}).
\label{oracle2}\\[-14pt]
\end{enumerate}
The algorithm can easily be generalized to incrementally compute any
polytope $P$
if the oracle associated with the problem satisfies property~\eqref{oracle1}; if
it satisfies 
also property~\eqref{oracle2}, then the computation can be done in
$O(|P|+|P^H|)$ oracle calls, where 
$|P|$, $|P^H|$ denotes the number of vertices and number of facets of $P$,
respectively.
For example, if the described oracle returns $\pi(\phi_T)$ instead of
$\pi(\rho_T)$, it can be used to compute orthogonal projections of secondary
polytopes.

The algorithm readily yields an approximate variant:
for each supporting hyperplane $H$,
we use its normal $w$ to compute $v=${VTX}$(\A,w,\pi)$.
Instead of computing a convex hull, now simply
take the hyperplane parallel to $H$ through $v$. The set of
these hyperplanes defines a polytope $Q_{o}\supseteq \varPi$, i.e.\ an
\textit{outer approximation} of $\varPi$.
In particular, at every step of the algorithm, $Q$ and $Q_{o}$ are 
an inner and an outer approximation of $\varPi$, respectively.
Thus, we have an approximation algorithm by stopping Algorithm~\ref{AlgComputeP}
when $\text{vol}(Q)/\text{vol}(Q_{o})$ achieves a user-defined threshold.
Then, $\text{vol}(Q)/\text{vol}(\varPi)$ is bounded by the same threshold.
Implementing this algorithm yields a speedup of up to 25 times 
(Section~\ref{Simplement}).
It is clear that vol$(Q)$ is available by our incremental convex hull algorithm.
However, vol$(Q_o)$ is the critical step; we plan to examine algorithms
that update (exactly or approximately) this volume.  

When all hyperplanes of $Q$ are checked,
knowledge of legal hyperplanes accelerates subsequent computations of $\Q$,
although it does not affect its worst-case complexity.
Specifically, it allows us to avoid checking legal facets against
new vertices.  
 
\section{Hashing of Determinants} \label{Shasheddets}

This section discusses methods to avoid duplication of computations by
exploiting the nature of the determinants appearing in the inner loop of
our algorithm.  
Our algorithm computes many regular triangulations,
which are typically dominated by the computation of determinants.
A similar technique, using dynamic determinant computations, is used to
improve determinantal predicates in incremental convex hull
computations~\citess{FP_ESA12}.

Consider the $2n\times|\A|$ matrix with the points of $\A$ as columns.
Define $P$ as the extension of this matrix by adding lifting values
$\widehat{w}$ as the last row.
We use the Laplace (or cofactor) expansion along the last row for computing
the determinant of the square submatrix formed by any $2n+1$ columns of
\(P\); without loss of generality, we assume these are the first $2n+1$
columns \(a_1,\ldots,a_{2n+1}\).
Let  $(1,\ldots,2n+1)\setminus i$ be the vector
resulting from removing the \(i\)-th element from the vector 
$(1,\ldots,2n+1)$ and let \(P_{(1,\ldots,2n+1) \setminus i}\)
be the \((2n) \times (2n)\) matrix obtained from the
\(2n\) elements of the columns whose indices are in 
$(1,\ldots,2n+1)\setminus i$.

The Orientation predicate is the sign of the determinant of
$P^{hom}_{(1,\ldots,2n+2)}$,
constructed by columns
\(a_1,\ldots,a_{2n+2}\) and adding \(\vec{1}\in \RR^{2n+2}\) as the last row.
Computing a regular subdivision is a long sequence of such predicates,
varying $a_i$'s on each step.
We expand along the next-to-last row, which contains the lifting values,
and compute the determinants
\( | P_{(1,\ldots,2n+2)\setminus i} |\)
for \(i\in \{1,\ldots,2n+2\}\).
Another predicate is Volume, used by {VTX}$(\cdot)$.
It equals the determinant of $P^{hom}_{(1,\ldots,2n+1)}$,
constructed by columns \(a_1,\ldots,a_{2n+1}\)
and replacing the last row of the matrix by \(\vec{1}\in \RR^{2n+1}\).

\begin{example}\label{ExamHash}
Consider the polynomials
$f_0:=c_{00}-c_{01}x_1x_2+c_{02}x_2$,
$f_1:=c_{10}-c_{11}x_1x_2^2+c_{12}x_2^2$ and
$f_2:=c_{20}-c_{21}x_1^2+c_{22}x_2$
and the lifting vector $\widehat{w}$ yielding the matrix \(P\).
\setlength{\tabcolsep}{0pt} 
\newcommand{\w}{6mm} 
$$P=
\begin{tabular}{|m{.5mm} >{\centering}m{\w} >{\centering}m{\w}
>{\centering}m{\w} >{\centering}m{\w} >{\centering}m{\w} >{\centering}m{\w}
>{\centering}m{\w} >{\centering}m{\w} >{\centering}m{\w} m{.5mm}| l}

\cline{1-1}
\cline{11-11}

& $0$ &  $0$ &  $0$ & $1$ & $1$ & $2$ & $0$ & $0$ & $0$ & &
        \multirow{2}{*}{ \Big\} \text{\small support coordinates}}\\
& $0$ & $0$ & $0$ & $1$ & $2$ & $0$ &  $1$ & $2$ & $1$ & \\

& \color{BurntOrange} $0$ & \color{BurntOrange} $1$ & \color{BurntOrange}$0$ & \color{BurntOrange}$0$ &
        \color{BurntOrange}$1$ &\color{BurntOrange} $0$ &
        \color{BurntOrange}$0$ & \color{BurntOrange}$1$ & \color{BurntOrange}$0$ & &
        \multirow{2}{*}{ \Big\} \text{\small Cayley trick coordinates}}\\

& \color{BurntOrange}$0$ & \color{BurntOrange}$0$ & \color{BurntOrange}$1$ & \color{BurntOrange}$0$ &
        \color{BurntOrange}$0$ & \color{BurntOrange}$1$ & 
        \color{BurntOrange}$0$ & \color{BurntOrange}$0$ & \color{BurntOrange}$1$ & \\

& \color{MidnightBlue}$w_1$ & \color{MidnightBlue}$w_2$ & \color{MidnightBlue}$w_3$ & \color{MidnightBlue}$0$ &
        \color{MidnightBlue}$0$ & \color{MidnightBlue}$0$ & \color{MidnightBlue}$0$ & \color{MidnightBlue}$0$ &
        \color{MidnightBlue}$0$ & &
        \multirow{1}{*}{ \} \text{\small $\widehat{w}$}}\\

\cline{1-1}
\cline{11-11}

\end{tabular}
$$
We reduce the computations of predicates to computations of minors of the
matrix obtained from deleting the last row of $P$.
Computing an Orientation predicate using Laplace expansion consists of
computing $\binom{6}{4}=15$ minors. On the other hand, if we compute
$|P^{hom}_{(1,2,3,4,5,6)}|$, the computation of
$|P^{hom}_{(1,2,3,4,5,7)}|$ requires
the computation of only
$\binom{6}{4}-\binom{5}{4}=10$ new minors.
More interestingly, when given a new lifting $\widehat{w'}$, we
compute $|P'~^{hom}_{(1,2,3,4,5,6)}|$
without computing any new minors.
\end{example}

Our contribution consists in maintaining a hash table with the computed
minors, which will be reused at subsequent steps of the algorithm.  
We store all minors of sizes between \(2\) and \(2n\).
For Orientation, they are independent of \(w\) and once computed they
are stored in the hash table.
The main advantage of our scheme is that, for a new ${w}$, the only change
in $P$ are $m$ (nonzero) coordinates in the last row, hence 
computing the new determinants
can be done by reusing hashed minors.
This also saves time from matrix constructions.

Laplace expansion computation of a matrix of size \(n\) has
complexity\linebreak
\(O(n) \sum_{i=1}^n L_i\), where \(L_i\) is the cost of computing the
\(i\)-th minor. \(L_i\) equals \(1\) when the \(i\)-th minor was
precomputed; otherwise, it is bounded by \(O\big((n-1)!\big)\).
This allows us to formulate the following Lemma.
\begin{lemma}
Using hashing of determinants, the complexity of the Orientation and Volume
predicates is $O(n)$ and $O(1)$, respectively, if all minors have
already been computed.
\end{lemma}

Many determinant algorithms modify the input matrix;
this makes necessary to create a new matrix and introduces a constant
overhead on each minor computation.
Computing with Laplace expansion, while
hashing the minors of smaller size, performs better than
state-of-the-art algorithms, in practice.
Experiments in Section~\ref{Simplement} show that our algorithm with hashed
determinants outperforms the version without hash.
For \(m=3\) and \(m=4\), we experimentally observed that the speedup factor
is between 18 and 100; Figure~\ref{fig:gfan_hash} illustrates the second case.

The drawback of hashing determinants is the amount of storage, which 
is in \(O(n!)\).
The hash table can be cleared at any moment to limit memory consumption, at
the cost of dropping all previously computed minors. Finding a policy
to clear the hash table according to the number of times each minor was
used would decrease the memory consumption, while keeping running times
low.
Exploring different heuristics, such as using a LRU (least recently used)
cache, to choose which minors to drop when freeing memory will be an
interesting research subject.

It is possible to exploit the structure of the above $(2n)\times(2n)$ minor
matrices. Let $M$ be such a matrix, with
columns corresponding to points of $A_0,\dots,A_n$.
After column permutations, we split $M$ into four $n\times n$
submatrices $A,B,D, I$, where $I$ is the identity matrix 
and $D$ has at most one $1$ in each column.
This follows from the fact that the bottom half of every column in $M$
has at most one $1$ and the last $n$ rows of $M$ contain at least
one $1$ each, unless $\det M=0$, which is easily checked.
Now, $\det M=\pm \det(B-AD)$, with $AD$ constructed in $O(n)$.
Hence, the computation of $(2n)\times(2n)$ minors
is asymptotically equal to computing 
an $n\times n$ determinant. 
This only decreases the constant within the asymptotic bound.
A simple implementation of this idea is not faster than Laplace expansion in
the dimensions that we currently focus.  However, this idea should be valuable
in higher dimensions.

\section{Implementation and Experiments} \label{Simplement}

We implemented Algorithm~\ref{AlgComputeP} in C++ to compute $\varPi$; our code
can be obtained from 
\begin{center}
 \url{http://respol.sourceforge.net}. 
\end{center}
All timings shown in this section were obtained on an Intel Core i5-2400
$3.1$GHz, with $6$MB L2 cache and $8$GB RAM, running 64-bit Debian
GNU/Linux.

Our implementation, {\tt respol}, relies on CGAL, using mainly a
preliminary version of package {\tt triangulation}~\citess{BoiDevHor09}, 
for both regular
triangulations, as well as for the V- and H-representation of $\varPi$.
As for hashing determinants, we looked for a hashing function, that takes
as input a vector of integers and returns an integer, which
minimizes collisions.
We considered many different hash functions, including some variations of
the well-known FNV hash~\citess{fnv}.
We obtained the best results with the implementation of Boost
Hash~\citess{boosthash}, which shows fewer collisions than the other tested
functions.
We clear the hash table when it contains $10^6$ minors. This gives a good
tradeoff between efficiency and memory consumption. Last column of
Table~\ref{tab:CHcompare} shows that the memory consumption of our
algorithm is related to $|A|$ and $\dim(\varPi)$.

We start our experiments by comparing four state-\-of-\-the-\-art exact
convex hull packages: {\tt triangulation} implementing Ref.~\refcite{CMS93}
and 
{\tt beneath-\-and-\-beyond (bb)} in {\tt polymake}~\citess{GaJo02};
double description implemented in {\tt cdd}~\citess{cddFuku};
and {\tt lrs} implementing reverse search~\citess{Avis98lrs}.
We compute $\varPi$, actually extending the work in Ref.~\refcite{AvBrSe97}
for the new class of polytopes $\varPi$.
The {\tt triangulation} package was shown to be faster in computing Delaunay
triangulations in
$\le 6$ dimensions~\citess{BoiDevHor09}.  
The other three packages are run through {\tt polymake}, where we have ignored
the time to load the data.
We test all packages in an {offline version}. We first compute the
V-representation of $\varPi$ using our implementation and then we give this as
an input to the convex hull packages that compute the H-representation of
$\varPi$. Moreover, we test {\tt triangulation}
by inserting points in the order that Algorithm~\ref{AlgComputeP} computes them,
while improving the point location of these points since we know by the
execution of Algorithm~\ref{AlgComputeP} one facet to be  removed (online version).
The experiments show that {\tt triangulation} and {\tt bb}
are faster than {\tt lrs}, which outperforms {\tt cdd}. 
Furthermore, the online version of {\tt triangulation} is $2.5$ times
faster than its offline counterpart due to faster point location
(Table~\ref{tab:CHcompare}, Figure~\ref{fig:CHcompare}).

\begin{table*}[t]\footnotesize \centering
\begin{tabular}{@{}crr|rr|rrrr|r}
\multirow{2}{*}{$m$} & 
\multicolumn{1}{c}{\multirow{2}{*}{$|\A|$}} &
\multicolumn{1}{c|}{\# of $\varPi$} &
\multicolumn{6}{c|}{time (seconds)} &
\multicolumn{1}{c}{{\tt respol}} \\
\cline{4-9} & &
\multicolumn{1}{c|}{vertices} &
\multicolumn{1}{c}{{\tt respol}} &
\multicolumn{1}{c|}{{\tt tr/on}} &
\multicolumn{1}{c}{{\tt tr/off}} &
\multicolumn{1}{c}{{\tt bb}} &
\multicolumn{1}{c}{{\tt cdd}} &
\multicolumn{1}{c|}{{\tt lrs}} &
\multicolumn{1}{c}{Mb}\\
\hline
3 & 2490 & 318 & 85.03 & 0.07 & 0.10 & 0.07 & 1.20 & 0.10 & 37 \\
4 & 27 & 830 & 15.92 & 0.71 & 1.08  & 0.50 & 26.85 & 3.12 & 46 \\ 
4 & 37 & 2852 & 97.82 & 2.85 & 3.91  & 2.29 & 335.23 & 39.41 & 64 \\ 
5 & 15 & 510 & 11.25 & 2.31 & 5.57  & 1.22 & 47.87 & 6.65 & 44 \\ 
5 & 18 & 2584 & 102.46 & 13.31 & 34.25 & 9.58 & 2332.63 & 215.22 & 88 \\
5 & 24 & 35768 & 4610.31 & 238.76 & 577.47 & 339.05 & $>1$hr & $>1$hr & 360 \\
6 & 15 & 985 & 102.62 & 20.51 & 61.56 & 28.22 & 610.39 & 146.83 & 2868 \\
6 & 19 & 23066 & 6556.42 & 1191.80 & 2754.30 & $>1$hr & $>1$hr & $>1$hr &
        6693 \\
7 & 12 & 249 & 18.12 & 7.55 & 23.95 & 4.99 & 6.09 & 11.95 & 114 \\
7 & 17 & 500 & 302.61 & 267.01 & 614.34 & 603.12 & 10495.14 & 358.79 &
        5258 \\
\end{tabular}
\caption{Total time and memory consumption of our code ({\tt respol}) and
time comparison of online version 
of {\tt triangulation (tr/on)} and offline versions of all convex hull
packages for
computing the H-representation of $\varPi$.}
\label{tab:CHcompare} \end{table*}

\begin{figure*}[t] \centering
  \subfigure[]{\label{fig:ch}\includegraphics[width=0.48\textwidth]
{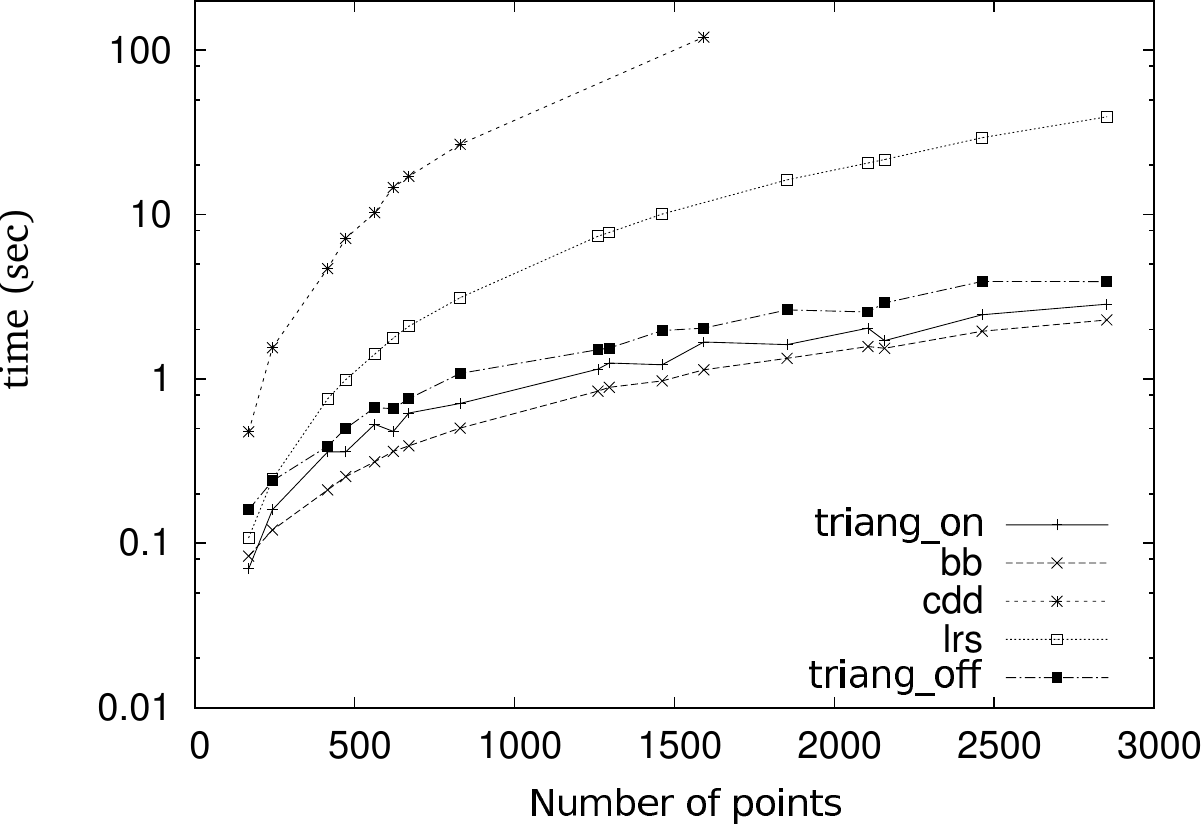}}
\subfigure[]{\label{fig:ch5}\includegraphics[width=0.48\textwidth]
{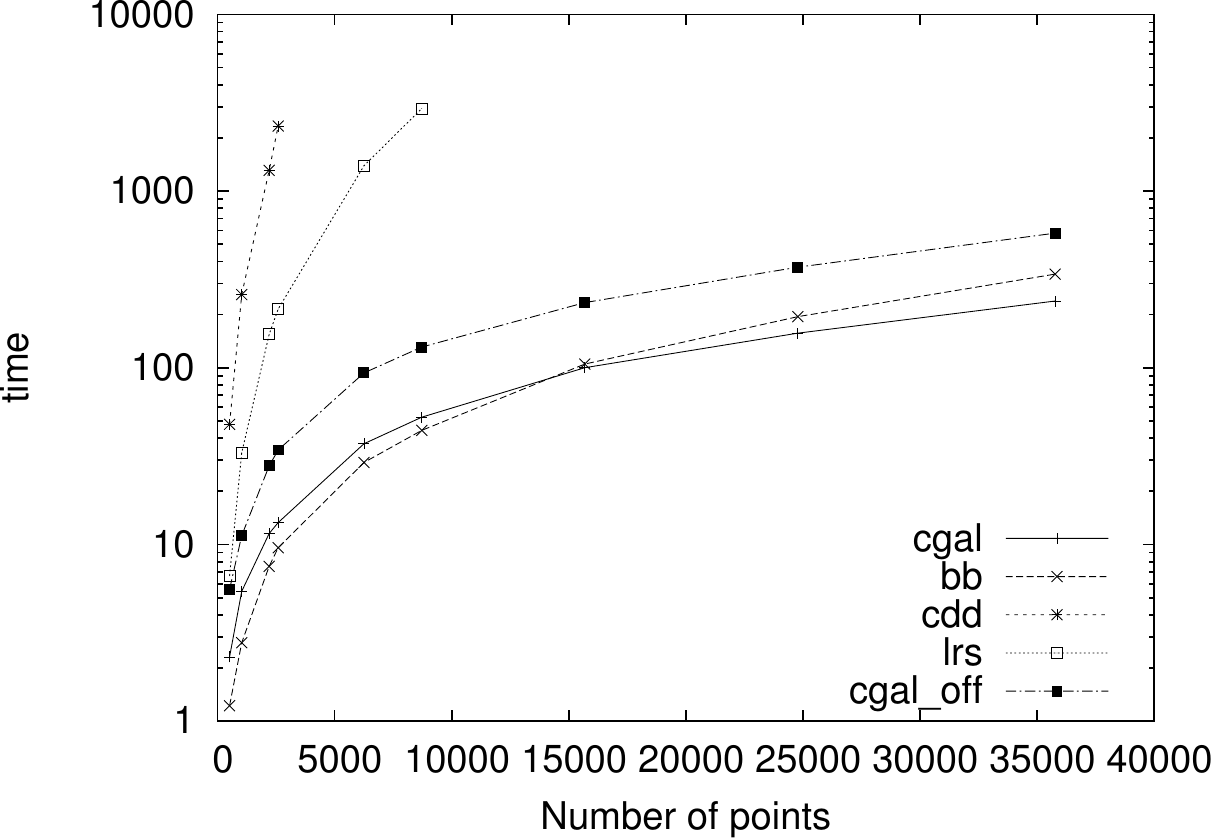}}
  \caption{Comparison of convex hull packages for  $4$-dimensional (a) and
$5$-dimensional (b) $\varPi$. 
{\tt triang\_on}/{\tt triang\_off} are the online/offline versions of {\tt
triangulation} package (y-axis is in logarithmic scale).
\label{fig:CHcompare}}
\end{figure*}

A {\it placing triangulation} of a set of points is a triangulation produced by
the Beneath-and-Beyond convex hull algorithm for some ordering of the points.
That is, the algorithm places the points in the triangulation with respect to
the ordering. Each point which is going to be placed, 
is connected to all
visible faces of the current triangulation resulting to the construction of new
cells. 
An advantage of {\tt triangulation} is that it maintains a placing
triangulation of a polytope in $\RR^d$ by storing the $0,1,d-1,d$-dimensional
cells of the triangulation. This is useful when the oracle VTX$(\A,w,\pi)$ needs
to refine the regular subdivision of $\A$ which is obtained by projecting the
upper hull of the lifted pointset $\A^{\widehat{w}}$
(Section~\ref{Sproject}).
In fact this refinement is attained by a placing triangulation, i.e., by
computing the projection of the upper hull of the placing triangulation of
$\A^{\widehat{w}}$.
This is implemented in two steps:\\[-14pt]
\begin{enumerate}[Step 1.]\itemsep0pt
 \item compute the placing triangulation $T_0$ of the last $|\A|-m$ points
with a random insertion order as described in~Ref.~\refcite{BoiDevHor09}
(they all have height zero),
 \item place the first $m$ points of $\A^{\widehat{w}}$ 
 in $T_0$ with a random insertion order~\citess{BoiDevHor09}.
\end{enumerate}
Step~1 is performed only once at the beginning of the algorithm,
whereas Step~2 is performed every time we check a new $w$.
The order of placing the points in Step~2 only matters if $w$ is not generic;
otherwise, $w$ already produces a triangulation of the $m$ points,
so any placing order results in this triangulation.

This is the implemented method; although different from the perturbation
in the proof of Lemma~\ref{Lpointonboundary}, it is more efficient
because of the reuse of triangulation $T_0$ in Step~1 above.
Moreover, our experiments show that it
always validates the two conditions in Section~\ref{Sproject}.

We can formulate this 2-step construction using a single lifting.
Let $c>0$ be a sufficiently large constant,
$a_i \in \A,\, q_i \in \RR$, $q_i > c\, q_{i+1},$ for $i=1,\ldots,|\A|$.
Define lifting $h: \A \rightarrow \RR^2$, where  $h(a_i)=(w_i,q_i),$ for
$i=1,\ldots,m$, and $h(a_i)=(0,q_i)$, for $i=m+1,\ldots,|\A|$.
Then, projecting the upper hull of $\A^h$ to $\RR^{2n}$ yields the
triangulation of $\A$ obtained by the 2-step construction.

Fixing the dimension of the triangulation at compile time
results in $<1\%$ speedup. We also tested a kernel that uses the
filtering technique based on interval arithmetic from Ref.~\refcite{BBP98}
 with a similar time speedup.
On the other hand, {\tt triangulation} is expected to implement
incremental high-dimensional regular triangulations with respect to a lifting,
faster
than the above method~\citess{Devi11perso}.
Moreover, we use a modified version of \texttt{triangulation} in order to
benefit from our hashing scheme.
Therefore, all cells of the triangulated facets of $\varPi$ have the same
normal vector and we use a structure ({\tt STL set}) to maintain the set of
unique normal vectors, thus computing only one regular triangulation
per triangulated facet of $\varPi$.

\begin{figure*}[t]
  \centering
\subfigure[]{\label{fig:ures}\includegraphics[width=0.48\textwidth]
{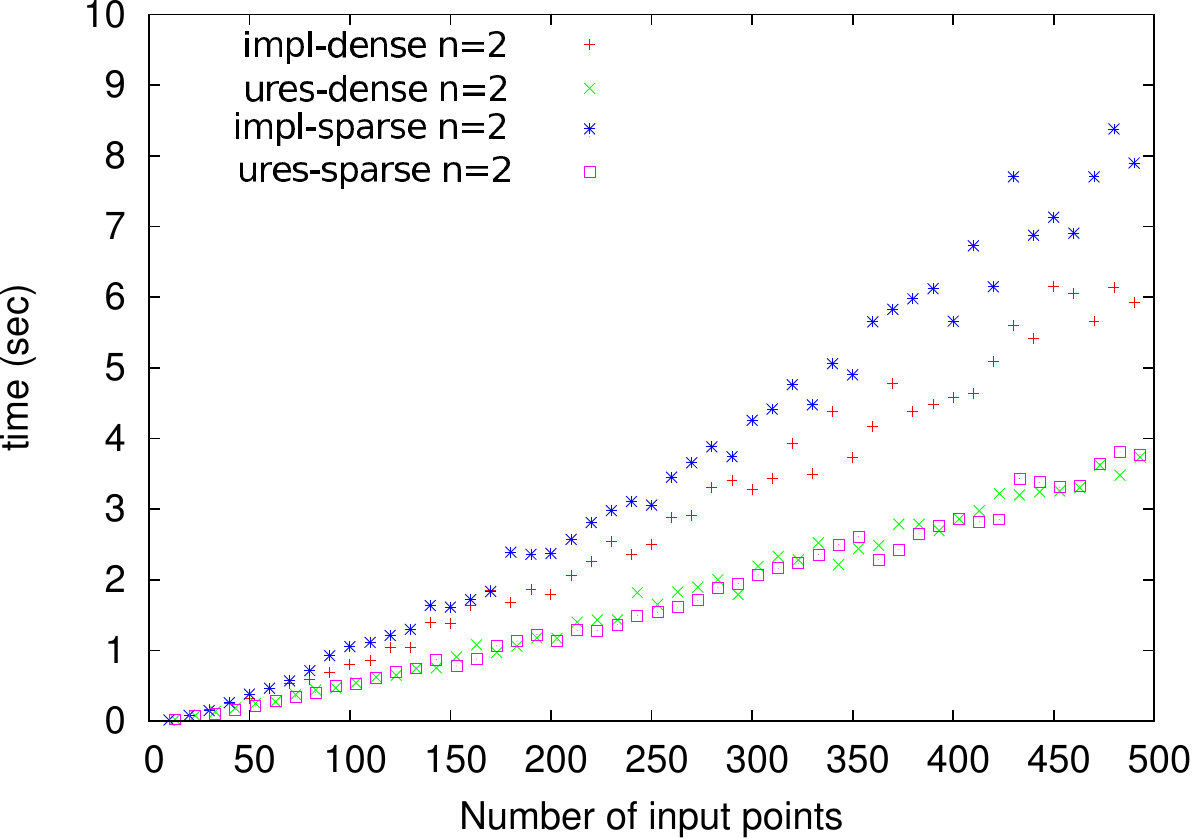}}
\subfigure[]{\label{fig:gfan_hash}\includegraphics[width=0.48\textwidth]
{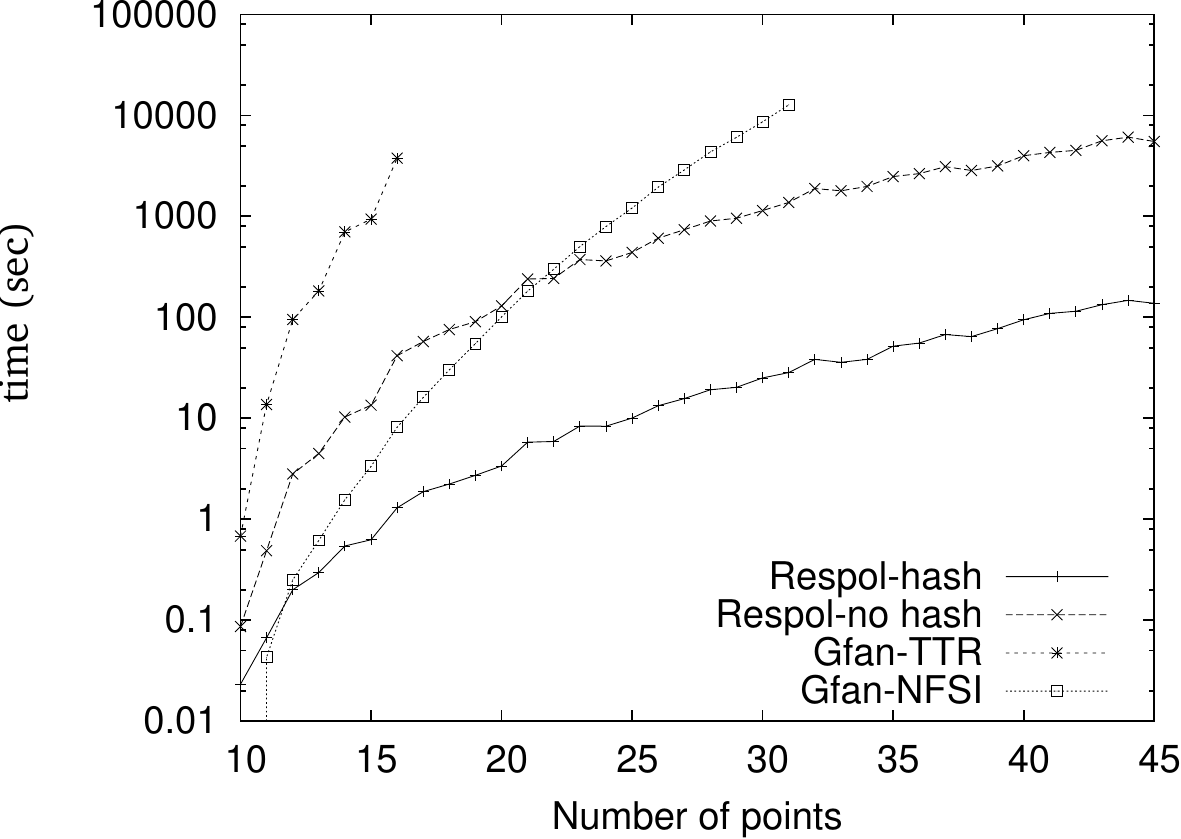}}
\subfigure[]{\label{fig:in}\includegraphics[width=0.48\textwidth]
{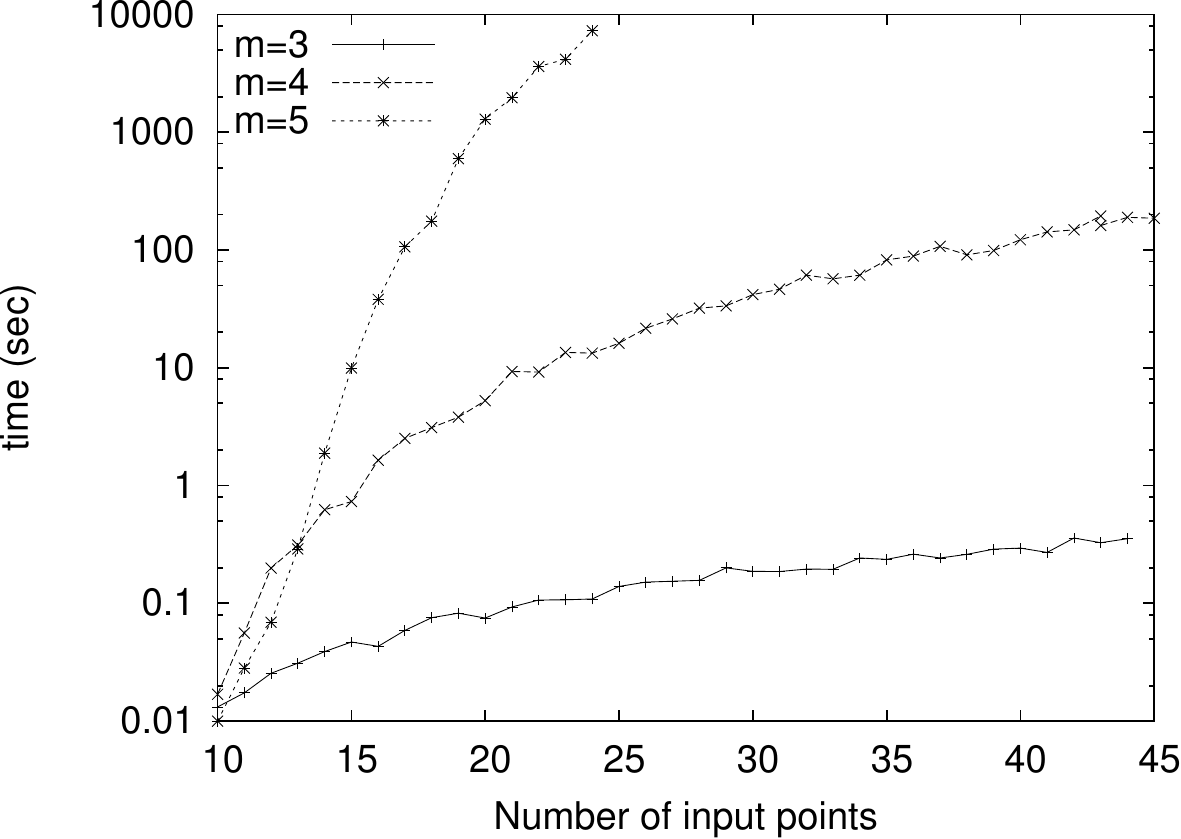}}\quad
\subfigure[]{\label{fig:output}\includegraphics[width=0.48\textwidth]
{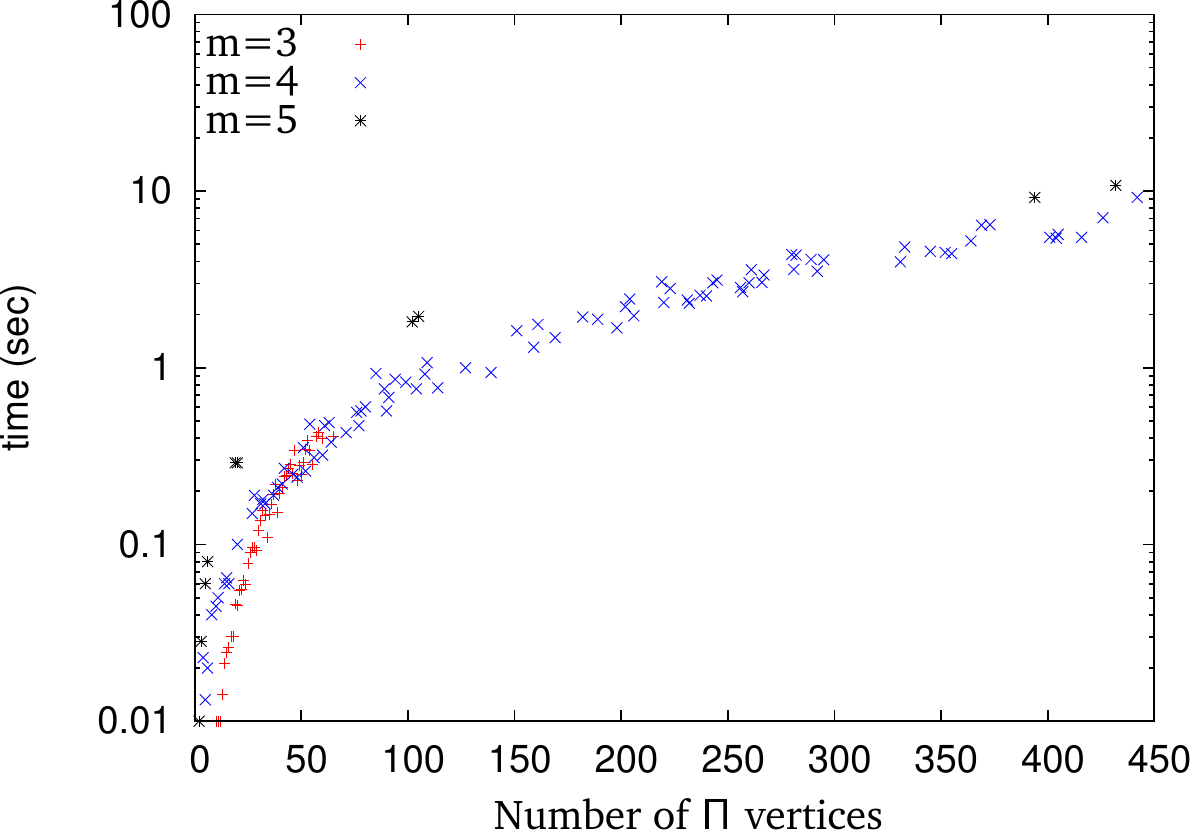}}
\caption{
(a) Implicitization and $u$-resultants for $n=2, m=3$;
(b) Comparison of \texttt{respol} (hashing and not hashing
determinants) and \texttt{Gfan} (traversing tropical resultants and
computing normal fan from stable intersection) for $m=4$;
(c) Performance of Alg.~\ref{AlgComputeP} for $m=3,4,5$ as a function of
input;
(d) Performance of Alg.~\ref{AlgComputeP} as a function of its output;
y-axes in (b), (c), (d) are in logarithmic scale.
}
\end{figure*}

We perform an \textit{experimental analysis} of our algorithm.  
We design experiments parameterized on:
the total number of input points $|\A|$, 
the dimension $n$ of pointsets $A_i$, and the dimension of projection $m$.
First, we examine our algorithm on random inputs for implicitization and
$u$-resultants, where $m=n+1$, while varying $|\A|, n$.
We fix $\delta\in\NN$ and select random points on the $\delta$-simplex
to generate dense inputs, and points on the $(\delta/2)$-cube
to generate sparse inputs.
For {\it implicitization} the projection coordinates correspond to
point $a_{i1}=(0,\ldots,0)\in A_i$.
For $n=2$ the problem corresponds to implicitizing surfaces:
when $|\A|<60$, we compute the polytopes in $<1$sec 
(Figure~\ref{fig:ures}).
When computing the {\it $u$-resultant} polytope, 
the projection coordinates correspond to 
$A_0=\{(1,\ldots,0),\ldots,(0,\dots,1)\}$. 
For $n=2$, 
when $|\A|<500$, we compute the polytopes in $<1$sec 
(Figure~\ref{fig:ures}).

By using the {\it hashing determinants} scheme we gain a $18\times$ speed\-up 
when $n=2,\ m=3$.
For $m=4$ we gain a larger speedup; we computed in $<2$min an
instance where $|\A|=37$ and would take $>1$hr to compute
otherwise.
Thus, when the dimension and $|\A|$ becomes larger, this method allows our
algorithm to compute instances of the problem that would be intractable
otherwise, as shown for \(n=3,\ m=4\)  
(Figure~\ref{fig:gfan_hash}).

We confirm experimentally the \textit{output-sensitivity} of our algorithm. 
First, our algorithm always computes vertices of $\varPi$ either to extend
$\varPi$ or to legalize a facet.
We experimentally show that our algorithm has, for fixed $m$, a
subexponential behaviour with respect to both input and output
(Figure~\ref{fig:in},~\ref{fig:output}) and its output is
subexponential with respect to the input. 

\renewcommand{\tabcolsep}{0.165cm}
\begin{table}[t]\scriptsize
\begin{tabular}{@{}rrrr@{ }|rrrr@{ }|r@{ }r@{ }r@{ }r@{ }r@{ }r@{}}
\multicolumn{4}{c|}{\# cells in triangulation} &
        \multicolumn{4}{c|}{time (sec)} &
        \multicolumn{6}{c}{\multirow{2}{*}{f-vector of $\varPi$}}\\\cline{1-8}
$\mu$ & $\sigma$ & min & max & $\mu$ & $\sigma$ & min & max \\\hline
4781 & 154 & 4560 & 5087 & 0.35 & 0.01 & 0.34 & 0.38 & 
& & 449 & 1405 & 1438 & 482\\
16966 & 407 & 16223 & 17598 & 1.51 & 0.03 & 1.45 & 1.56 & 
& & 1412 & 4498 & 4705 & 1619\\\hline
18229 & 935 & 16668 & 20058 & 1.92 & 0.10 & 1.77 & 2.11 & 
& 432 & 1974 & 3121 & 2082 & 505\\
563838 & 6325 & 548206 & 578873 & 99 & 1.62 & 93.84 & 103.07 & 
& 9678 & 43569 & 71004 & 50170 & 13059\\\hline
289847 & 15788 & 264473 & 318976 & 69 & 4.88 & 61.67 & 77.31 & 
1308 & 7576 & 16137 & 16324 & 7959 & 1504\\
400552 & 14424 & 374149 & 426476 & 96.5 & 4.91 & 88.86 & 107.12 & 
1680 & 9740 & 21022 & 21719 & 10890 & 2133\\
\end{tabular}
\caption{Typical f-vectors of projections of resultant polytopes and the size of
their triangulations. We perform  20 runs with random
insertion order of vertices for each polytope and report the minimum, maximum,
average value $\mu$ and the standard deviation $\sigma$ for the number of cells
and the runtime.}
\label{tbl:triang_size}
\end{table}

\begin{figure}[t]
  \centering
\subfigure[]{\label{fig:rand}\includegraphics[width=0.48\textwidth]
{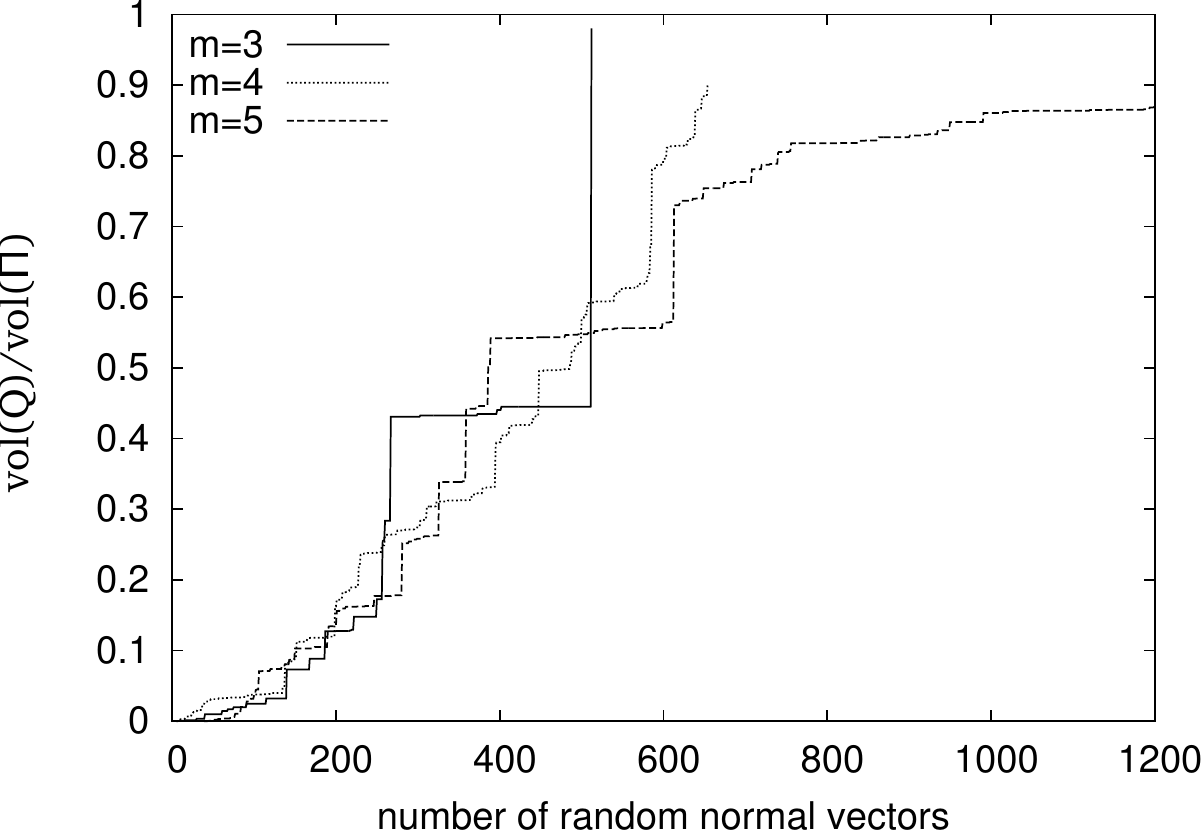}}
\subfigure[]{\label{fig:cells}\includegraphics[width=0.49\textwidth]
{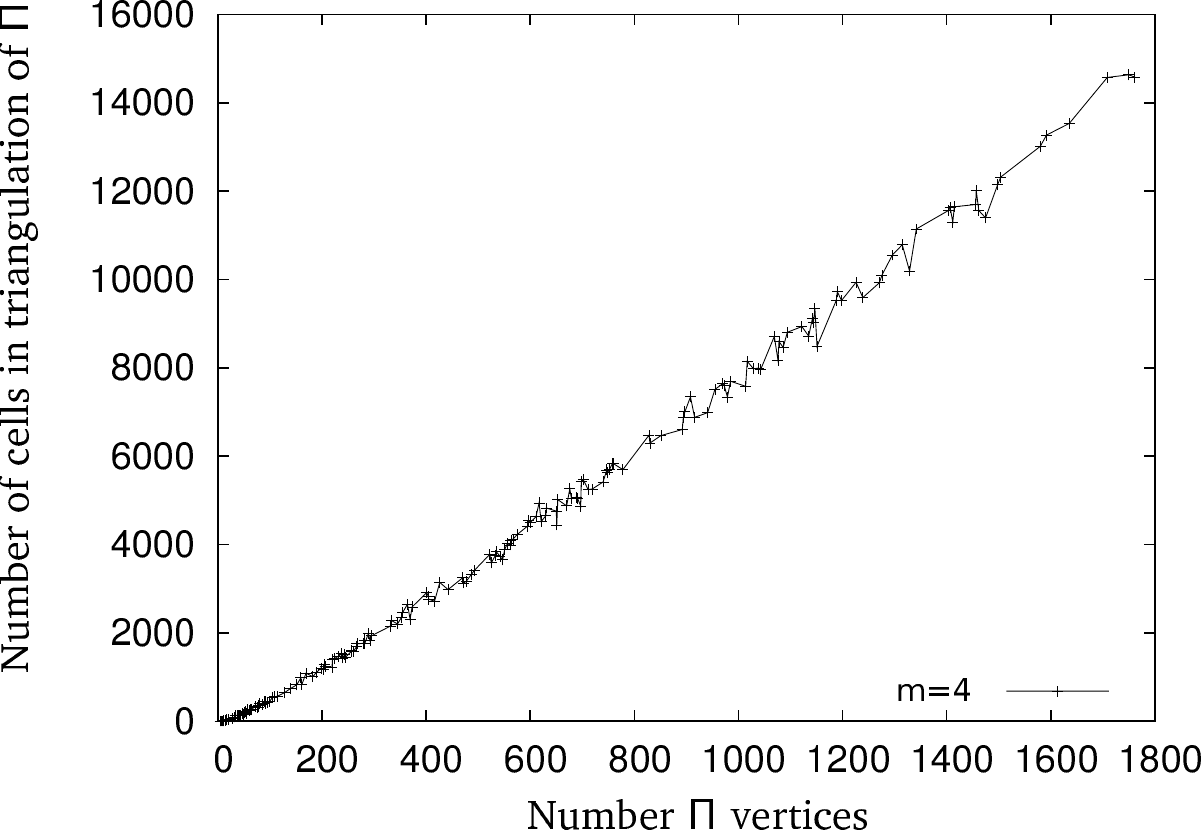}}
\caption{
(a) $\mbox{vol}(Q)/\mbox{vol}(\varPi)$ as a function of the number of
random normal vectors used to compute $Q$;
(b) The size of the triangulation of $\varPi$ as a function of the
output of Alg.~\ref{AlgComputeP}.
}
\label{fig:2}
\end{figure}

As the complexity analysis (Theorem~\ref{Ttotalcomplexity})
indicates, the runtime of the algorithm depends on the
size of the constructed placing triangulation of $\varPi$. The size
of the placing triangulation depends on the ordering of the inserted points. We
perform experiments on the effect of the inserting order to the size of the
triangulation as well as the running time of the computation of the 
triangulation~(Table~\ref{tbl:triang_size}). These sizes as well as the
runtimes vary in a very narrow range. Thus, the insertion order is not crucial
in both the runtime and the space of our algorithm.
Further experiments in $4$-dimensional $N(R)$ show
that the size of the input bounds polynomially the
size of the triangulation of the output (Figure~\ref{fig:cells}) which explains
the efficiency of our algorithm in this dimension.

We explore the \textit{limits} of our implementation. By bounding
runtime to $<2$hr, we compute instances of $5$-, $6$-, $7$-di\-men\-sional
$\varPi$ with $35$K, $23$K, $500$ vertices,
respectively\ (Table~\ref{tab:CHcompare}).

We also compare with the implementation of Ref.~\refcite{JensenYu11}, which
is based on \texttt{Gfan} library.
They develop two algorithms to compute projections of $N(\R)$.
Assuming $\R$ defines a hypersurface, their methods compute
a union of (possibly overlapping) cones, along with their
multiplicities, see Theorem~2.9 of Ref.~\refcite{JensenYu11}.
From this intermediate result they construct the normal cones to the resultant
vertices.

\begin{table}[h!]\footnotesize \centering \begin{tabular}{c|lll|llllll}
examples in Ref.~\refcite{JensenYu11}
                              &  a  &  b  &  c &  d &  e &  f &  g &  h & i \\
\hline
\(|\A|\)                      & 12  & 12  & 15 & 12 & 12 & 16 & 27 & 16 & 20\\
\( m  \)                      & 12  & 12  & 15 &  6 &  7 &  9 &  3 &  4 &  5\\
\( n  \)                      &  3  &  2  &  4 &  2 &  2 &  3 &  2 &  3 &  4\\
\texttt{Gfan}(secs$^*$)
                              &1.40 & 6   & 55 &0.70&1.30&798 &0.40&2.60&184 \\
\texttt{respol}(secs)  & 1.40&18.41 &99.90&0.26&1.24& 934&0.02&0.96&292.01\\
\end{tabular}
\caption{Comparison of our implementation with \texttt{Gfan}.
$^*$ Timings for \texttt{Gfan} as reported in Ref.~\protect\refcite{JensenYu11}.
\label{TComparisonGFAN}}
\end{table}

We compare with the best timings of {\tt Gfan} methods using the examples
and timings of Ref.~\refcite{JensenYu11} (Table~\ref{TComparisonGFAN}).
Our method is faster in examples (d), (e), (g), (h) where $m<7$, is
competitive (up to $2$ times slower) in (a) where $m=|\A|=12$ and (i) where
$m=5,|\A|=20$ and slower in (b), (c), (f) where $m\geq12$. The bottleneck of our
implementation, that makes it slower when the dimension of the projection $m$ is
high, is the incremental convex hull construction in $\RR^m$. 
Moreover, since our implementation considers that $N(\R)$ lies in
$\RR^{|\A|}$ instead of $\RR^{|\A|-2n-1}$, (see also the discussion on the
homogeneities of $\R$ in Section~\ref{Scombinatorics}),
it cannot  take advantage of the fact that
$\dim(N(\R))$ could be less than $m$ when $|\A|-2n-1 < m < |\A|$. 
This is the case in examples (b), (c) and
(f).
On the other hand, we run extensive experiments for $n=3$, considering
implicitization, where $m=4$ and
our method, with and without using hashing, is much faster than
any of the two algorithms based on \texttt{Gfan} (Figure~\ref{fig:gfan_hash}).
However, for $n=4,\ m=5$ the beta version of \texttt{Gfan} used in our experiments was not
stable and always crashed when \(|\A|>13\).

\renewcommand{\tabcolsep}{0.3cm} \begin{table*}[t]\footnotesize
\centering
\begin{tabular}{@{} c@{}c @{\hspace{0.2cm}}
|@{\hspace{0.5cm}}  c  @{\hspace{0.2cm}}  |  
@{\hspace{0.5cm}} rrrrrr  @{}}
\multirow{2}{*}{{ input}}
& & m & 3 & 3 & 4 & 4 & 5 & 5\\ 
& & $|\A|$& 200 & 490& 20 & 30& 17& 20\\\hline
\multirow{2}{*}{{ approximation}} &
& \# of $Q$ vertices&15 & 11& 63 & 121 & $>10$hr &$>10$hr\\
& & $\mbox{vol}(Q)/\mbox{vol}(\varPi)$& 0.96& 0.95& 0.93& 0.94 & $>10$hr
&$>10$hr\\
\multirow{2}{*}{{ algorithm}}
& & $\mbox{vol}(Q_{o})/\mbox{vol}(\varPi)$&1.02& 1.03 & 1.04& 1.03& $>10$hr
&$>10$hr\\
& & time (sec)& 0.15& 0.22& 0.37& 1.42& $>10$hr &$>10$hr\\\hline
\multirow{2}{*}{{ uniformly }}
&
& $|Q|$& 34& 45& 123 & 207&  228& 257\\ 
& & random vectors& 606 & 576& 613& 646& 977& 924\\
 \multirow{2}{*}{{ random}}
& & $\mbox{vol}(Q)/\mbox{vol}(\varPi)$& 0.93& 0.99& 0.94 & 0.90& 0.90& 0.90\\
& & time (sec)& 5.61& 12.78& 1.10&4.73& 8.41& 16.90\\\hline
\multirow{1}{*}{{ exact}}
&
 & \# of $\varPi$
vertices& 98 & 133& 416& 1296& 1674& 5093\\
\multirow{1}{*}{{ algorithm}}
& & time (sec)& 2.03& 5.87& 3.72& 25.97 &51.54& 239.96\\
\end{tabular}
\caption{Results on experiments computing $Q, \Qo$ using the
approximation algorithm and the random vectors procedure; we stop the
approximation algorithm when
$\mbox{vol}(Q)/\mbox{vol}(Q_{o})>0.9$; the results with random vectors
are the average
values over $10$ independent experiments; ``$>10$hr'' indicates 
computation of $\mbox{vol}(Q_{o})$ was interrupted after $10$hr. 
\label{randQ}}
\end{table*}

We analyze the computation of inner and
outer \textit{approximations} $Q$ and $\Qo$.
We test the variant of Section~\ref{Sproject} by
stopping it when $\mbox{vol}(Q)/\mbox{vol}(\Qo)>0.9$. In the
experiments, the number of $Q$ vertices is $<15\%$ of the 
$\varPi$ vertices, thus there is a speedup of up to $25$ times over
the exact algorithm at the largest instances.
The approximation of the volume is very satisfactory:
$\mbox{vol}(\Qo)/\mbox{vol}(\varPi)<1.04$ and
$\mbox{vol}(Q)/\mbox{vol}(\varPi) > 0.93$ 
for the tested instances (Table~\ref{randQ}).
The bottleneck here is the computation of
vol$(\Qo)$, where
$\Qo$ is given in H-representation: the runtime explodes for $m\ge 5$.
We use 
{\tt polymake} in every step to
compute vol$(\Qo)$ because we are lacking of an implementation that, given a
polytope $P$ in H-representation, its volume and a halfspace $H$, computes the
volume of the intersection of $P$ and $H$.
Note that we do not include this computation time in the reported time.
Our current work considers ways to extend these observations to a polynomial
time approximation algorithm for the volume and the polytope itself when the
latter is given by an optimization oracle, as is the case here.

Next, we study procedures that compute only the V-rep\-re\-sen\-ta\-tion of
$Q$.
For this, we count 
how many \textit{random vectors} uniformly distributed on the
$m$-dimensional sphere are needed to obtain
$\mbox{vol}(Q)/\mbox{vol}(\varPi)>0.9$. 
This procedure runs up to $10$ times faster than the exact algorithm
(Table~\ref{randQ}). 
Figure~\ref{fig:rand} illustrates the convergence of
$\mbox{vol}(Q)/\mbox{vol}(\varPi)$ to the threshold value $0.9$ in
typical $3,4,5$-dimensional examples. 
The basic drawback of this method is that it does not provide guarantees for
$\mbox{vol}(Q)/\mbox{vol}(\varPi)$ because we do not have sufficient
{\em a priori} information on $\varPi$.
These experiments also illustrate the extent in which
the normal vectors required to deterministically construct $\varPi$
are uniformly distributed over the sphere.

\section{Future work}

One algorithm that should be experimentally evaluated is the following. 
We perform a search over the vertices of $\Sigma(A)$, that is, we build a search
tree with flips as edges. We keep a set with the extreme vertices with respect
to a given projection. Each computed vertex that is not extreme in the above set is
discarded and no flips are executed on it, i.e. the search tree is pruned in
this vertex. The search procedure could be the algorithm of TOPCOM or the one
presented in Ref.~\refcite{MicVer99} which builds a search
tree in some equivalence classes of $\Sigma(A)$. The main advantage of this
algorithm is that it does not involve a convex hull computation. On the other
hand, it is not output-sensitive with respect to 
the number of vertices of the resultant polytope;
its complexity depends on the number of vertices on the \emph{silhouette}
of $\Sigma(A)$, with respect to a given projection and those that 
are connected by an edge with them.

As shown, {\tt polymake}'s convex hull algorithm is competitive,
thus one may use it for implementing our algorithm. 
On the other hand, {\tt triangulation} is expected to include fast
enumeration of all regular triangulations for a given (non generic)
lifting, in which case $\varPi$ may be extended by more
than one (coplanar) vertices.

Our proposed algorithm uses an incremental convex hull algorithm and it is
known that any such algorithm has a worst-case super-polynomial \emph{total time
complexity}~\citess{Bremner} in
the number of input points and output facets.
The basic open question that this paper raises is whether there is a polynomial
total time algorithm for $\varPi$ or even for the set of its vertices.

\section{Acknowledgments}
All authors were partially supported from project
``Computational Geometric Learning'', which acknowledges the
financial support of the Future and Emerging Technologies (FET)
programme within the 7th Framework Programme for research of
the European Commission, under FET-Open grant number: 255827.
Most of the work was done while C.~Konaxis and L.~Pe{\~n}aranda
were at the University of Athens.
C.~Konaxis' research leading to these results has also received
funding from the European Union's 
Seventh Framework Programme (FP7-REGPOT-2009-1) under grant agreement n\textsuperscript{o} 245749.
We thank O.~Devillers and S.~Hornus for discussions on
\texttt{triangulation}, and A.\ Jensen and J.~Yu for discussions 
and for sending us a beta version of their code.

\bibliographystyle{unsrt}
\bibliography{algebra,emiris,geometry,bibliography}

\begin{thebibliography}{10}

\bibitem{EmKaKoLB}
I.Z.\ Emiris, T.\ Kalinka, C.\ Konaxis, and T.~Luu Ba.
\newblock Implicitization of curves and (hyper)surfaces using predicted
  support.
\newblock {\em Theor. Comp. Science, Special Issue on Symbolic \& Numeric
  Computing}, 479(0):81--98, 2013.

\bibitem{StuYu08}
B.~Sturmfels and J.~Yu.
\newblock Tropical implicitization and mixed fiber polytopes.
\newblock In {\em Software for Algebraic Geometry}, volume 148 of {\em IMA
  Volumes in Math.\ \& its Applic.}, pages 111--131. Springer, New York, 2008.

\bibitem{RambTOPCOM}
J.~Rambau.
\newblock {TOPCOM}: Triangulations of point configurations and oriented
  matroids.
\newblock In {\em Proc. Intern. Congress Math. Software}, pages 330--340, 2002.

\bibitem{EKKL12spm}
I.Z. Emiris, T.~Kalinka, C.~Konaxis, and T.~Luu Ba.
\newblock Sparse implicitization by interpolation: Characterizing non-exactness
  and an application to computing discriminants.
\newblock {\em J.\ Computer Aided Design}, 45:252--261, 2013.
\newblock Special Issue on Symposium Solid \& Phys.\ Modeling 2012 (Dijon,
  France).

\bibitem{Rincon12}
F.~Rinc{\'o}n.
\newblock Computing tropical linear spaces.
\newblock In {\em J.\ Symbolic Computation}, volume~51, pages 86--98, 2013.

\bibitem{CGAL}
{CGAL}: Computational geometry algorithms library.
\newblock \url{http://www.cgal.org}.

\bibitem{JensenYu11}
A.~Jensen and J.~Yu.
\newblock Computing tropical resultants.
\newblock {\em arXiv:math.AG/1109.2368v1}, 2011.

\bibitem{GKZ}
I.M. Gelfand, M.M. Kapranov, and A.V. Zelevinsky.
\newblock {\em Discriminants, Resultants and Multidimensional Determinants}.
\newblock Birkh\"{a}user, Boston, 1994.

\bibitem{St94}
B.\ Sturmfels.
\newblock On the {Newton} polytope of the resultant.
\newblock {\em J.\ Algebraic Combin.}, 3:207--236, 1994.

\bibitem{DEF12}
A.~Dickenstein, I.Z. Emiris, and V.~Fisikopoulos.
\newblock Combinatorics of 4-dimensional resultant polytopes.
\newblock In {\em Proc.\ ACM Intern.\ Symp.\ on Symbolic \& Algebraic Comput.},
  2013.
\newblock (to appear).

\bibitem{discrim_vol}
S.Yu. Orevkov.
\newblock {The volume of the Newton polytope of a discriminant.}
\newblock {\em Russ. Math. Surv.}, 54(5):1033--1034, 1999.

\bibitem{IMTI02}
H.~Imai, T.~Masada, F.~Takeuchi, and K.~Imai.
\newblock Enumerating triangulations in general dimensions.
\newblock {\em Intern.\ J.\ Comput.\ Geom.\ Appl.}, 12(6):455--480, 2002.

\bibitem{MicCoo00}
T.\ Michiels and R.~Cools.
\newblock Decomposing the secondary {Cayley} polytope.
\newblock {\em Discr.\ Comput.\ Geometry}, 23:367--380, 2000.

\bibitem{MicVer99}
T.\ Michiels and J.\ Verschelde.
\newblock Enumerating regular mixed-cell configurations.
\newblock {\em Discr.\ Comput.\ Geometry}, 21(4):569--579, 1999.

\bibitem{Hug06}
P.~Huggins.
\newblock ib4e: A software framework for parametrizing specialized lp problems.
\newblock In A.~Iglesias and N.~Takayama, editors, {\em Mathematical Software
  (ICMS 2006)}, volume 4151 of {\em Lecture Notes in Computer Science}, pages
  245--247. Springer, Berlin, 2006.

\bibitem{KaVi05}
{E}. {K}altofen and {G}. {V}illard.
\newblock On the complexity of computing determinants.
\newblock {\em Computational Complexity}, 13:91--130, 2005.

\bibitem{BEPP99}
H.~Br\"{o}nnimann, {I.Z.} Emiris, V.~Pan, and S.~Pion.
\newblock Sign determination in {R}esidue {N}umber {S}ystems.
\newblock {\em Theor.\ Comp.\ Science, Spec.\ Issue on Real Numbers \&
  Computers}, 210(1):173--197, 1999.

\bibitem{DGGGHKSTV}
J.-G. Dumas, T.~Gautier, M.~Giesbrecht, P.~Giorgi, B.~Hovinen, E.~Kaltofen,
  B.~D. Saunders, W.~J. Turner, and G.~Villard.
\newblock Linbox: A generic library for exact linear algebra.
\newblock In {\em Proc. Intern. Congress Math. Software}, pages 40--50,
  Beijing, 2002.

\bibitem{eigenweb}
G.~Guennebaud, B.~Jacob, et~al.
\newblock Eigen~v3.
\newblock \url{http://eigen.tuxfamily.org}, 2010.

\bibitem{EFKP12}
{I}.{Z}. Emiris, {V}. {F}isikopoulos, {C}. {K}onaxis, and {L}. {P}e{\~n}aranda.
\newblock An output-sensitive algorithm for computing projections of resultant
  polytopes.
\newblock In {\em Proc.\ Annual ACM Symp.\ Computational Geometry}, pages
  179--188, 2012.

\bibitem{DeLRamSan}
J.A.~De Loera, J.~Rambau, and F.~Santos.
\newblock {\em Triangulations: Structures for Algorithms and Applications},
  volume~25 of {\em Algorithms and Computation in Mathematics}.
\newblock Springer, 2010.

\bibitem{Ziegler}
G.M. Ziegler.
\newblock {\em Lectures on Polytopes}.
\newblock Springer, 1995.

\bibitem{CLO2}
D.~Cox, J.~Little, and D.~O'Shea.
\newblock {\em Using Algebraic Geometry}.
\newblock Number 185 in GTM. Springer, New York, 2nd edition, 2005.

\bibitem{Kap91}
M.M. Kapranov.
\newblock Characterization of {A}-discriminantal hypersurfaces in terms of
  logarithmic {Gauss} map.
\newblock {\em Math. Annalen}, 290:277--285, 1991.

\bibitem{BaPoRo}
S.~Basu, R.~Pollack, and M.-F. Roy.
\newblock {\em Algorithms in real algebraic geometry}.
\newblock Springer-Verlag, Berlin, 2003.

\bibitem{BoiDevHor09}
J.-D. Boissonnat, O.~Devillers, and S.~Hornus.
\newblock Incremental construction of the {Delaunay} triangulation and the
  {Delaunay} graph in medium dimension.
\newblock In {\em Proc.\ Annual ACM Symp.\ Computational Geometry}, pages
  208--216, 2009.

\bibitem{Josw03bb}
M.~Joswig.
\newblock Beneath-and-beyond revisited.
\newblock In M.~Joswig and N.~Takayama, editors, {\em Algebra, Geometry, and
  Software Systems}, Mathematics and Visualization. Springer, Berlin, 2003.

\bibitem{FP_ESA12}
V.~Fisikopoulos and L.~Pe{\~n}aranda.
\newblock Faster geometric algorithms via dynamic determinant computation.
\newblock In {\em Proc.\ 20th Europ. Symp. Algorithms}, pages 443--454, 2012.

\bibitem{fnv}
G.~Fowler, L.C. Noll, and P.~Vo.
\newblock {F}owler/{N}oll/{V}o ({FNV}) hash algorithm.
\newblock \url{http://www.isthe.com/chongo/tech/comp/fnv}, 1991.

\bibitem{boosthash}
D.~James.
\newblock Boost functional library.
\newblock \url{http://www.boost.org/libs/functional/hash}, 2008.

\bibitem{CMS93}
K.L. Clarkson, K.~Mehlhorn, and R.~Seidel.
\newblock Four results on randomized incremental constructions.
\newblock {\em Comput.\ Geom.: Theory \& Appl.}, 3:185--121, 1993.

\bibitem{GaJo02}
E.\ Gawrilow and M.\ Joswig.
\newblock Polymake: an approach to modular software design in computational
  geometry.
\newblock In {\em Proc.\ Annual ACM Symp.\ Computational Geometry}, pages
  222--231. ACM Press, 2001.

\bibitem{cddFuku}
K.~Fukuda.
\newblock {cdd} and {cdd+} {Home Page}.
\newblock ETH Z\a{"}urich. \url{http://www.ifor.math.ethz.ch/~fukuda/cdd_home}, 2008.

\bibitem{Avis98lrs}
D.~Avis.
\newblock lrs: A revised implementation of the reverse search vertex
  enumeration algorithm.
\newblock In {\em Polytopes: Combinatorics \& Computation}, volume~29 of {\em
  Oberwolfach Seminars}, pages 177--198. Birkh{\"a}user, 2000.

\bibitem{AvBrSe97}
D.\ Avis, D.\ Bremner, and R.~Seidel.
\newblock How good are convex hull algorithms?
\newblock {\em Comput.\ Geom.: Theory \& Appl.}, 7:265--301, 1997.

\bibitem{BBP98}
H.~Br\"{o}nnimann, C.~Burnikel, and S.~Pion.
\newblock Interval arithmetic yields efficient dynamic filters for
  computational geometry.
\newblock In {\em Proc.\ Annual ACM Symp.\ Computational Geometry}, pages
  165--174, New York, 1998. ACM.

\bibitem{Devi11perso}
O.\ Devillers, 2011.
\newblock Personal communication.

\bibitem{Bremner}
D.~Bremner.
\newblock Incremental convex hull algorithms are not output sensitive.
\newblock In {\em Proc. 7th Intern. Symp. Algorithms and Comput.}, pages
  26--35, London, UK, 1996. Springer.

\end{thebibliography}

\end{document}